%% file: paper20140114.tex
\documentclass[12pt]{article}
\usepackage{authblk}
\usepackage{multirow}
\usepackage{lscape}
\usepackage[authoryear,round]{natbib}
\usepackage{float}
\usepackage{amsthm}
\input{preamble}

\def\1{\mathbf{1}}

\author[1]{Ngoc Mai Tran}
\author[2]{Maria Osipenko}
\author[3]{Wolfgang Karl H\"{a}rdle}
\affil[1]{Department of Mathematics, University of Texas at Austin, USA.}
\affil[2]{Collaborative Research Center 649: {\it Economic Risk}, Humboldt-Universit\"{a}t zu Berlin, Berlin, Germany.}
\affil[3]{C.A.S.E.- Center for Applied Statistics \& Economics, Humboldt-Universit\"{a}t zu Berlin, Berlin, Germany.\newline
Lee Kong Chian School of Business, Singapore Management University,  Singapore.}

\title{Principal Component Analysis in an Asymmetric Norm \thanks{This research was supported by Deutsche Forschungsgemeinschaft through the SFB 649 "Economic Risk".
Ngoc Tran was also supported by DARPA (HR0011-12-1-0011) and an award from the Simons Foundation (\# 197982 to The University of Texas at Austin)
}}
\begin{document}
\date{}
\maketitle
\begin{abstract}Principal component analysis (PCA) is a widely used dimension reduction tool in the analysis of many kind of high-dimensional data. It is used in signal processing, mechanical engineering, psychometrics, and other fields under different names. It still bears the same mathematical idea: the decomposition of variation of a high dimensional object into uncorrelated factors or components. However, in many of the above applications, one is interested in capturing the tail variables of the data rather than variation around the mean. Such applications include weather related event curves, expected shortfalls, and speeding analysis among others. These are all high dimensional tail objects which one would like to study in a PCA fashion. The tail character though requires to do the dimension reduction in an asymmetric norm rather than the classical $L_2$-type orthogonal projection. We develop an analogue of PCA in an asymmetric norm. These norms cover both quantiles and expectiles, another tail event measure. The difficulty is that there is no natural basis, no `principal components', to the $k$-dimensional subspace found. We propose two definitions of principal components and provide algorithms based on iterative least squares. We prove upper bounds on their convergence times, and compare their performances in a simulation study. We apply the algorithms to a Chinese weather dataset with a view to weather derivative pricing.\\

\noindent\textbf{Keywords:} principal components; asymmetric norm; dimension reduction; quantile; expectile.\\

\noindent\textbf{JEL Classification:} C38, C61, C63.
\end{abstract}
\section{Introduction}
When data come as curves without known functional form, the statistician faces immediately the need for dimension reduction. The conventional and widely used tool for such high dimensional curve data is principal component analysis (PCA). The basic principle of this technique is to treat the curves as random variations around a mean curve, and then orthogonalize the covariance operator into eigenfunctions and corresponding (random) loadings. The focus of this principle is on studying the variation around a mean curve. Loadings on (interpretable) eigenfunctions would then represent specific variations around the average. PCA or more generally functional PCA (FPCA) has been successfully applied in many fields such as gene expression measurements, weather, natural hazard, and environment studies, demographics, etc, see \citet{j04}, \citet{crknsa}, 
and \citet{chenmuel}. 
One of the first applications is in \citet{rs05}. They considered temperature curves recorded daily over a year at multiple stations in an area. The premise is that there are only a few principal components influencing the average temperature, and that the temperature curve from each station is well-approximated on average by a specific linear combinations of these factors. PCA approximates the mean of the data by a nested sequence of optimal subspaces of small dimensions. Thus the optimal subspace of dimension $k$ comes with a natural basis, consisting of uncorrelated random curves (vectors), the principal components, playing the role of the factors aforementioned. Due to the nested structure of the optimal subspaces, one can compute the first few components using a greedy algorithm. The first principal component can be computed efficiently using iterative partial least squares.

In many of the above applications, one is not only interested in the variation around an average curve, but rather in features of the data that are expressible as scale (variance) or tail related functional data. In pricing of financial products where volatility is relevant, for example, the variation of the scale of risk factors is at the core of fair pricing. If one would like to construct weather derivatives or forecasts for the above FPCA example on temperature curves, one needs not only to know the variation across stations, but also the changing scale of the temperature curves, \citet{cadi}, \citet{benth}, and \citet{haelop}. In climatological science, one is interested in the extremes of certain natural phenomena like drought or rainfall. A tail indicator like a quantile of a conditional distribution when indexed by an explanatory variable also constitutes a curve. Therefore, such a quantile curve collection may also be treated in an FPCA context. Yet another tail-describing curve is the expectile function. Like the quantile curve, it can be represented via a solution with respect to an asymmetric norm. Expectiles have as well numerous application areas, especially in the calculation of risk measures of a financial asset or a portfolio. \citet{tailor2008} shows how a widely accepted risk measure such as expected shortfall can be assessed via expectiles. \citet{kuan} apply this tail measure in an autoregressive risk management context.

In this paper, we develop an analogue of PCA for quantiles and expectiles. The later, proposed by \citet{np}, is an analogue of the mean for quantiles. The quantile to level $\tau$ of a distribution with cdf $F$, assuming $F$ is invertible, is defined as $q_\tau = F^{-1}(\tau)$. It is also the solution to the following optimization problem: 
$$ q_\tau = \arg\min_{q \in \R} \rE\|X - q\|_{\tau, 1} $$
where $X$ is a random variable in $\R$ with distribution $F$, and $\|x\|_{\tau,\alpha}^\alpha$ is the asymmetric norm:
\begin{equation} \label{eqn:alpha1}
\|x\|_{\tau,\alpha}^\alpha = |I(x \leq 0) - \tau||x|^\alpha, \hspace{1em} \alpha = 1.
\end{equation}
Given $X_i \sim F, i = 1, \ldots, n$, one may formulate the estimation of the unknown quantile in a location model:
\begin{equation}\label{eqn:noisy}
X_i = q_\tau + \varepsilon_i,
\end{equation}
with the $\tau$-quantile of the cdf of $\varepsilon$ being zero. A natural estimate of $q_\tau$ in (\ref{eqn:noisy}) is:
\begin{equation}\label{eqn:qtau1}
\hat{q}_\tau = \arg\min_{q \in \R}\sum_{i=1}^n \|X_i-q\|_{\tau,1}.
\end{equation}
The estimator as written in (\ref{eqn:qtau1}) can be defined for $\R^p$-valued vectors, if the asymmetric norm is taken by applying (\ref{eqn:alpha1}) coordinatewise and then summing over the coordinates. Given this extension it can be used to analyse curves data when discretized on a regular grid as mensioned in \citet{kneip}.

Formulation (\ref{eqn:qtau1}) yields a statistical interpretation. In fact, if the noise $\varepsilon_i$ in (\ref{eqn:noisy}) follows a so-called asymmetric Laplace distribution $ALD(\tau)$, which has cdf proportional to the functional $\exp(-\|\cdot\|_{\tau,1})$, then (\ref{eqn:qtau1}) can be interpreted as a quasi maximum likelihood estimation of equation (\ref{eqn:noisy}). 
Putting $\alpha = 2$ in (\ref{eqn:alpha1}) yields, via (\ref{eqn:qtau1}), a quasi likelihood interpretation based on an asymmetric normal distribution. Both cases $\alpha = 1$ and $\alpha = 2$ for $\tau \neq 0.5$ are indicators for a certain tail index. This paper aims to shed some light on how to create suitable subspace decompositions for such collections of tail index curves. 

As noted in \citet{gzhh}, the first step in this problem corresponds to doing low-rank matrix approximation with weighted $L_1$ and $L_2$ norm, respectively, where the weights are sign-sensitive (see Section 1). Based on a proposal of \citet{schnabel} an iterative weighted least squares algorithm for expectiles is employed where the weights are updated in each iteration. This algorithm is guaranteed to converge, although not necessarily to the global minimum as we shall show below. Thus one can at least find a locally optimal $k$-dimensional subspace that best approximates a given quantile or expectile curve (vector). The difficulty is that the weight matrix is not of rank one, hence there is no natural basis, no `principal components', to the $k$-dimensional subspace found. While this is a known problem in weighted low-rank matrix approximation, see \citet{srebro}, this problem has not been addressed before. 

Furthermore, the definition of an optimal $\tau$-expectile subspace employed in \citet{gzhh} is not invariant under linear transformations of the data. That is, if one changes the basis of the data, the optimal $\tau$-expectile subspace in the new basis is not necessarily a linear transform of that expressed in the old basis. This means one has to fix a basis for the data before computing the optimal $\tau$-expectile subspace. This restricts the usefulness of this method to applications where there is a natural basis, such as in the Chinese weather dataset, where yearly temperature is expressed as a vector of 365 daily temperatures. Here one would be interested in capturing extreme daily temperature as opposed to extreme temperature expressed in a Fourier basis. However, in many other applications, invariance under change of basis is an important feature of PCA. 

The contributions of our paper is two fold. After defining the basic concepts in the next section \ref{not}, we, first, work with the formulation in \citet{gzhh} in section \ref{sec:buptd} and propose two natural bases, hence two definitions of principal components for the optimal subspace found. Second, in section \ref{sec:pec} we propose an alternative definition of principal components for quantiles and expectiles, closely related to the definition of principal directions for quantiles of \citet{fraimanlopez}. This definition satisfies many nice properties, such as invariance under translations and linear transformations of the data. In particular, it returns the usual PCA basis under elliptically symmetric distributions.  We then provide algorithms to compute the three versions of principal components aforementioned, based on iterative weighted least squares in section \ref{sec:alg}. We prove upper bounds on their convergence times in section \ref{ssec:bounds} and compare their performances in a simulation study in section \ref{sec:sim}. In section \ref{sec:app} of our paper, we show an application to a Chinese weather dataset with a view to pricing weather derivatives. The last section summarizes our findings.

\section{Quantiles and expectiles}\label{not}
\subsection{Definitions}
We now set up notations and recall the definitions of quantile and expectile. In the next two sections we specify the main optimization problems in $\R^p$. 

For $y \in \R^p$, define $y_+ \stackrel{\textup{def}}{=} \max(0, y)$, $y_- \stackrel{\textup{def}}{=} \max(0, -y)$ coordinatewise. For $\tau \in (0,1)$, let $\|\cdot\|_1$ denote the $L_1$-norm in $\R^p$, that is, $\|y\|_1 = \sum_{j=1}^p|y_j|$. Define the asymmetric $L_1$-norm in $\mathbb{R}^p$ via
$$ \|y\|_{\tau, 1} = \tau\|y_+\|_1 + (1-\tau)\|y_-\|_1 = \sum_{j=1}^p|y_j|\cdot \left\{ \tau I(y_j \geq 0) + (1-\tau)I(y_j < 0) \right\}. $$
Similarly, let $\|\cdot\|_2$ denote the $L_2$-norm in $\R^p$, $\|y\|^2_2 = \sum_{j=1}^p y_j^2$. Define the asymmetric $L_2$-norm in $\mathbb{R}^p$ via
$$ \|y\|^2_{\tau,2} = \tau\|y_+\|_2^2 + (1-\tau)\|y_-\|_2^2. $$
When $\tau = 1/2$, we recover a constant multiple of the $L_1$ and $L_2$-norms. These belong to the class of asymmetric norms with sign-sensitive weights, and have appeared in approximation theory, \citet{asym}. Some properties we use in this paper are the fact that these norms are convex, and their unit balls restricted to a given orthant in $\R^p$ are weighted simplices for the $\|\cdot\|_{\tau, 1}$ norm, and axis-aligned ellipsoids for the $\|\cdot\|_{\tau,2}$ norm. In other words, they coincide with the unit balls of axis-aligned weighted $L_1$ and $L_2$ norms.

Let $Y \in \R^p$ be a random variable with cdf $F$. The $\tau$-quantile $q_\tau(Y) \in \R^p$ of $F_Y$ is the solution to the following optimization problem
$$ q_\tau(Y) = \argmin_{q \in \R^p} \rE \| Y - q \|_{\tau, 1}. $$
Similarly, the $\tau$-expectile $e_\tau(Y) \in \R^p$ of $F_Y$ is the solution to
$$ e_\tau(Y) = \argmin_{e \in \R^p} \rE \| Y - e \|^2_{\tau, 2}. $$ 

Since the asymmetric $L_1$ and $L_2$ norms are convex, 
the solution exists and is unique, assuming that $\rE(Y)$ is finite. This definition guarantees that the $\tau$-quantile $q_\tau(Y)$ is unique even when the cdf $F$ is not invertible.

\subsection{Properties}
We collect some mathematical properties of quantiles and expectiles here. These will be useful for proving theorems in future sections. For convenience we shall suppress the dependence on $Y$ where possible. We shall state the next proposition for the one-dimensional case, that is, $Y \in \R$. Analogous results in higher dimensions hold coordinatewise.

\begin{prop}[Properties of expectile \citet{np}]\label{prop:exp.prop} Let $Y \in \R$ be a random variable. Let $F$ be its cdf, $G$ be its first partial moment, defined as
$$ G(x) = \int_{-\infty}^x u\, dF(u). $$ 
Assume that $G(x) < \infty$ for all $x \in \R$
\begin{itemize}
	\item For $\tau \in (0,1)$, $e_\tau(Y+t) = e_\tau(Y) + t$ for $t \in \R$. 
	\item For $\tau \in (0,1)$, 
$$ e_\tau(sY) = \left\{\begin{array}{ccc} se_\tau(Y) & \mbox{ for } & s \in \R, s > 0 \\
-se_{1-\tau}(Y) & \mbox{ for } & s \in \R, s < 0
\end{array} \right. $$
	\item $e_\tau = e_\tau(Y)$ is the $\tau$-quantile of the distribution function $T$, i.e. $\tau = T(e_\tau)$ where
\begin{align}T(x) = \frac{G(x) - x F(x)}{2\{G(x) - x F(x)\} + \{x - \int_{-\infty}^{\infty} u\, dF(u)\}}.\label{Tx}\end{align}
\end{itemize}
\end{prop}

\begin{cor} Suppose $Y \in \R^p$ has a symmetric distribution about 0, and it belongs to the location-scale family. Then for $\tau \in (0,1)$,
$ e_\tau(Y) = -e_{1-\tau}(Y) $, and 
$$ e_\tau(sY + t) = se_\tau(Y) + t $$
for all $s \in \R, t \in \R^p$. In particular, if $Y \sim \textup N(\mu, \sigma^2)$, $Z \sim \textup N(0,1)$, then
$$ e_\tau(Y) = |\sigma| e_\tau(Z) + \mu. $$
\end{cor}

Suppose the cdf $F$ is differentiable. Then $q_\tau(Y) = F^{-1}(\tau)$. Let $F_n^{-1}: (0,1) \to \R$ and $F^{-1}: (0,1) \to \R$ denote the empirical and population quantile function, respectively. A classical result of empirical process theory in \citet[\S 2]{pollard} states that 
$$\sqrt{n}(F_n^{-1} - F^{-1})(t) \stackrel{\mathcal L}{\to} \frac{-W^0}{F'(F^{-1})}(t), ~t\in(0,1)$$
where $W^0$ denotes the Brownian bridge on $[0,1]$, and the convergence takes place over the Skorokhod space $D([0,1])$. Now, the last point of Proposition \ref{prop:exp.prop} states that the expectile is indeed the quantile of a function $T$. Thus, one may suspect that the expectile process also satisfies a similar statement. Indeed, we now make this concrete.

\begin{thm}\label{thm:exp.process}
Let $F$ be a differentiable cdf which defines a distribution with mean zero, variance $\sigma^2$. Let $e: (0,1) \to \R, \tau \mapsto e_\tau$ be the expectile function. Let $F_n, e_n$ be the empirical versions. Then for any $0 < \delta < 1$,
$$ \sqrt{n}(e_{n} - e) \stackrel{\mathcal L}{\to} \mathcal{E}, $$
where the convergence takes place over $D([\delta, 1-\delta])$, $\mathcal{E}$ is a stochastic process on $[\delta, 1-\delta]$, whose marginals are normally distributed with mean $0$ and variance
\begin{equation}\label{eqn:variance.jkj}
 \textup{Var}\{\mathcal{E}(\tau)\} = \frac{\textup{\rE}\{\tau (Y-e_\tau)_+ + (1-\tau)(e_\tau-Y)_+\}^2}{\left[\tau\{1-F(e_\tau)\} + (1-\tau)F(e_\tau)\right]^2}
 \end{equation}
for $\tau \in [\delta, 1-\delta]$. 
 \end{thm}
For example, if $\tau = 1/2$, then $e_{1/2,n}$ and $e_{1/2}$ are just the empirical and population mean, and $\textup{Var}\{\mathcal{E}(1/2)\} = \sigma^2$. Thus we recover the classical central limit theorem.
We first give an overview of the proof. We shall prove convergence for the inverse process of $e_{\tau,n}$ and $e_{\tau}$, which is $T_n(e_\tau)$ and $T(e_\tau)$ as defined in Proposition \ref{prop:exp.prop}. Then we invoke the result of \citet{dossgill} to show that the expectile process itself must also converge to a stochastic process $\mathcal{E}_\tau$. Finally, to derive the marginal distribution of $\mathcal{E}_\tau$ with $e_\tau$ being the solution of a \emph{convex} optimization problem. Thus its asymptotic properties, in particular, its limit in distribution, can be derived involving a theorem of \citet{hjort2011asymptotics}. Applying the result of Hjort and Pollard can only give finite dimensional convergence of the process $\sqrt{n}(e_{\tau,n} - e_\tau)$. On the other hand, it is possible to derive Theorem \ref{thm:exp.process} using the result of Doss and Gill alone, however, the computation for the second moment of $\mathcal{E}_\tau$ is quite messy. Thus we choose to only derive the second moment properties of the process $\mathcal{E}_\tau$. 

\begin{proof}
Note that the inverse process of $e_{\tau,n}$ and $e_\tau$ are $T_n: \R \to [0,1], e_\tau \mapsto T_n(e_\tau)$ and $T: \R \to [0,1], e_\tau \mapsto T(e_\tau)$ as defined by (\ref{Tx}) in Proposition \ref{prop:exp.prop}. To be clear,
$$ T_n(x) = \frac{G_n(x) - xF_n(x)}{2\{G_n(x) - xF_n(x)\} + (x - \mu_n)}, $$
where $G_n(x) = \int_{-\infty}^x u\, dF_n(u)$ is the empirical version of $G$, and $\mu_n = \int_{-\infty}^\infty udF_n(u)$ is the empirical mean. By \citet{np}, the functions $T_n, T$ are both distribution functions, and thus they are non-decreasing cadlag functions. We claim that the stochastic processes $\sqrt{n}(T_n - T)$ converges to some stochastic process in the Skorokhod space $D([-\infty,\infty])$. Indeed, note that the processes $\sqrt{n}(G_n - G)$ and $\sqrt{n}(F_n - F)$ both converge to some process on $D([-\infty,\infty])$. Similarly, assuming $\mu = 0$, $\sqrt{n}\mu_n$ converges to the normal distribution with mean $0$, variance $\sigma^2$. The numerator of the fraction $\sqrt{n}\{T_n(x) - T(x)\}$ is
$$\sqrt{n}\{G_n(x) - G(x)\}x - \sqrt{n}\{F_n(x) - F(x)\}x^2 -\sqrt{n}\mu_n\{G(x) - xF(x)\}.$$
Since $F$ has finite second moment, $|G(x) - xF(x)|$ is uniformly bounded for large $x$. Thus the above expression converges in distribution uniformly in $x$. Now, the denominator of the fraction $T_n(x) - T(x)$  is
$$[2\{G(x) - xF(x)\} + x][2\{G_n(x) - xF_n(x)\} + (x - \mu_n)],$$
which converges a.s. for all $x$ to $[2\{G(x) - xF(x)\} + x]^2$, which is bounded away from $0$. Thus the process $\sqrt{n}(T_n - T)$ converges in $D([-\infty,\infty])$. \\
By \citet{dossgill}, this implies that the inverse processes $T_n^{-1} = e_n$, $T^{-1} = e$ must satisfy
$$ \sqrt{n}(e_n - e) \stackrel{\mathcal L}{\to} \mathcal{E} $$
where the convergence takes place over $D([\delta, 1-\delta])$, $\mathcal{E}$ is a stochastic process on $[\delta, 1-\delta]$. Finally, to derive the marginal distribution of $\mathcal{E}_{\tau}$ with $e_\tau$ being the solution of a \emph{convex} optimization problem with differentiable objective function. Thus the asymptotic properties of the empirical estimator $e_{\tau,n}$, and in particular, its limit in distribution, can be derived using the theorem of \citet[Theorem 2]{hjort2011asymptotics}. Explicitly, in their notation, fix $\tau \in [\delta, 1-\delta]$, and let $g_\tau(y,t) = \|y-t\|_{\tau,2}^2 = \tau(y - t)_+^2 + (1-\tau)(t-y)_+^2$ be our objective function. Differentiate with respect to $t$, we find
$$ g'_\tau(y,t) = 2\{-\tau(y-t)_+ + (1-\tau)(t-y)_+\}, \hspace{1em} g''_\tau(y,t) = 2\{\tau I(y \geq t) + (1-\tau)I({y < t})\}. $$
Define $K = \textup{\rE}g'_\tau(Y,t)^2 = 4\textup{\rE}\{\tau (Y-e_\tau)_+ + (1-\tau)(e_\tau-Y)_+\}^2$, and 
$$J = \textup{\rE}\{g''_\tau(Y,e_\tau)\} = \textup{\rE}2\{\tau I(y \geq t) + (1-\tau)I({y < t})\} = 2[\tau\{1-F(e_\tau)\} + (1-\tau)F(e_\tau)].$$
Now, since $g$ is a convex differentiable function, as $t \to e_\tau$, 
$$ \textup{\rE}\{g_\tau(Y,t) - g_\tau(Y,e_\tau)\} = \frac{1}{2}\textup{\rE}\{g''_\tau(Y,e_\tau)\}(e_\tau - t)^2 + {\scriptsize\text{$\mathcal{O}$}}(|t|^2).$$
Therefore, by \citet[Theorem 2]{hjort2011asymptotics}, $\sqrt{n}\{e_{\tau,n} - e_\tau\}$ converges to a normal distribution with mean $0$ and variance $J^{-1}KJ^{-1}$, which in our case simplifies to (\ref{eqn:variance.jkj}).
\end{proof}


\section{Principal components as error minimizers}\label{sec:buptd}
There are multiple, equivalent ways to define standard PCA, which generalize to different definitions of principal components for quantiles and expectiles. We focus on two formulations: minimizing the residual sum of squares, and maximizing the variance capture. 
\subsection{Review of PCA}
Suppose we observe $n$ vectors $Y_1, \ldots, Y_n \in \R^p$ with edf $F_n$. Write $Y$ for the $n \times p$ data matrix. PCA solves for the $k$-dimensional affine subspace that best approximates $Y_1, \ldots, Y_n$ in $L_2$-norm. In matrix terms, we are looking for the constant $m^\ast \in \R^p$ and the matrix $E^\ast_k$, the rank-$k$ matrix that best approximates $Y - \mathbf{1}(m^\ast)^\top$ in the Frobenius norm. That is,
\begin{equation}\label{eqn:optim.estar} 
(m^\ast_k,E^\ast_k) = \argmin_{m\in \R^p, E \in \R^{n \times p}: rank(E) = k} \|Y - \mathbf{1}m^\top - E \|^2_{1/2, 2}.
\end{equation}
As written, $m$ is not well-defined: if $(m,E)$ is a solution, then $(m+c,E-\mathbf{1}c^\top)$ is another equivalent solution for any $c$ in the column space of $E$. Geometrically, this means we can express the affine subspace $m+E$ with respect to any chosen point $m$. It is intuitive to choose $m$ to be the best constant in this affine subspace that approximates $Y$. By a least squares argument, the solution is $m_k^\ast = \rE(Y)$. That is, it is independent of $k$ and coincides with the best constant approximation to $Y$. Thus, it is sufficient to assume $\rE(Y) = m \equiv 0$, and consider the optimization problem in (\ref{eqn:optim.estar}) without the constant term.

Suppose $Y$ is full rank and the eigenvalues of its covariance matrix are all distinct. Again by least squares argument, for $1 \leq k < p$, the column space of $E_k^\ast$ is contained in the column space of $E_{k+1}^\ast$, and $E_{k+1}^\ast - E_k^\ast$ is the optimal rank-one approximation of $Y - E_k^\ast$. This has two implications. Firstly, there exists a natural basis for $E_k^\ast$. Indeed, there exists a unique ordered sequence of orthonormal vectors $v_1,v_2,\ldots,v_p \in \R^p$ such that $E_1^\ast = U_1V_1^\top$,$E_2^\ast=U_2V_2^\top$, and so on, where the columns of $V_k$ are the first $k$ $v_i$'s. The $v_i$'s are called the \emph{principal components}, or \emph{factors}. For fixed $k$, $V_k$ is the \emph{component}, or \emph{factor matrix}, and $U_k$ is the \emph{loading}. 

Secondly, a greedy algorithm reduces computing the components $v_1,v_2,\ldots$ to computing the first component, in other words, solving (\ref{eqn:optim.estar}) for $k=1$. Every rank-one matrix $E \in \R^{n \times p}$ has a unique decomposition $E = UV^\top$ for $U \in \R^{n \times 1}$, $V \in \R^{p \times 1}$ with $V^\top V = 1$. Thus, solving (\ref{eqn:optim.estar}) is equivalent to an unconstrained minimization problem over the pair of matrices $(U,V)$ with objective
$$ J(U,V) = \|Y - UV^\top\|_{1/2, 2}^2 = \sum_{i,j}(Y_{ij} - \sum_l U_{il}V_{jl})^2.$$
For fixed $U$, $J$ is a quadratic in the entries of $V$, and vice versa. Since all local minima of $J$ are global, see \citet{srebro}, $J$ can be efficiently minimized using an iterative least squares algorithm, leading to an efficient method for performing PCA for small $k$ in large datasets. 

\subsection{Analogues for expectiles}\label{sec:errmin}
We now generalize the above definition of PCA to handle expectiles. The quantiles case follows similarly, and algorithms for $L_1$ matrix factorization can also be adapted to this case. Recall that we are looking for the best $k$-dimensional affine subspace which minimizes the asymmetric $L_2$-norm. The analogue of (\ref{eqn:optim.estar}) is the following low-rank matrix approximation problem
\begin{equation}\label{eqn:optim.estar.tau} 
(m^\ast_k,E^\ast_k) = \argmin_{m\in \R^p, E \in \R^{n \times p}: rank(E) = k} \|Y - \mathbf{1}m^\top - E \|^2_{\tau, 2}.
\end{equation}
Again, we may define $m$ to be the best constant approximation to $Y$ on the affine subspace determined by $(m,E)$. For a fixed affine subspace, such a constant is unique, and is the coordinatewise $\tau$-expectile of the residuals $Y - E$. However, the expectile is not additive for $\tau \neq 1/2$. Thus in general, the column space of $E_k^\ast$ is not a subspace of the column space $E_{k+1}^\ast$, the constant $m_k^\ast$ depends on $k$, and is not equal to the $\tau$-expectile $e_\tau(Y)$. 

Let us fix $k$ and consider the problem of computing $m_k^\ast$ and $E_k^\ast$. Write a rank-$k$ matrix $E$ as $E = UV^\top$, where $U \in \R^{n \times k}, V \in \R^{p \times k}$. Adjoin $U$ with an all-1 column to form $\tilde{U}$, and adjoin $m$ to the corresponding column of $V$ to form $\tilde{V}$. Thus $\mathbf{1}m^\top + E = \tilde{U}\tilde{V}^\top$. Equation (\ref{eqn:optim.estar.tau}) is an unconstrained minimization problem over the pair of (adjoined) matrices $(\tilde{U}, \tilde{V})$ with minimization objective
$$ J(\tilde{U},\tilde{V},W) = \|Y - \tilde{U}\tilde{V}^\top\|_{\tau, 2}^2 = \sum_{i,j}w_{ij}(Y_{ij} - m_j - \sum_l U_{il}V_{jl})^2.$$
where the weights $w_{ij}$ are sign-dependent:
$w_{ij} = \tau$ if $Y_{ij} - m_j - \sum_l U_{il}V_{lk} > 0$, $w_{ij} = 1-\tau$ otherwise. 

This objective function is not jointly convex in $\tilde{U}$ and $\tilde{V}$. However, for fixed $\tilde{U}$, in each coordinate $ij$, it is the asymmetric $L_2$-norm of a linear combination in the entries of $\tilde{V}$, and hence convex. Similarly, $J$ is convex in $\tilde{U}$ for fixed $\tilde{V}$. Therefore, an iterative weighted least squares solution with weight update at each step is guaranteed to converge to a critical point of $J$ (cf. Proposition \ref{prop:laws.gradient.descent}). This algorithm (cf Algorithm \ref{alg:laws}) is called asymmetric weighted least squares (LAWS), see \citet{np} and \citet{schnabel}. While there are local minima, we find that the algorithm often finds the global minimum quite quickly, supporting similar observations in the literature for fixed weight matrix $[w_{ij}]$, as in \citet{srebro}.

For $k > 1$, the decomposition $E = UV^\top$ is not unique: for any $k \times k$ matrix $R$, the matrix $(UR, V(R^\top)^{-1})$ is another equivalent factorization.  To specify a unique solution we need a choice for $V$. This is one of the unaddressed issues in \citet{gzhh}, and certainly a key difficulty. While there are algorithms to solve for $(m_k^\ast, E_k^\ast)$ for fixed $k$, there is no natural basis for $E_k^\ast$ which reveals information on $E_j^\ast$ for $j < k$. Hence, we do not have a direct analogue for principal components for $\tau$-expectiles.

To furnish a principal components basis for $E^\ast_k$ based on LAWS, we propose two algorithms: TopDown and BottomUp. These are two definitions, described as algorithms, which output is a nested sequence of subspaces, each approximating $E_j^\ast$ for $j = 1, \ldots, k$. They lead to two different definitions of principal components. 

\begin{defn} Given data $Y \in \R^{n \times p}$ and an integer $k \geq 1$, the first $k$ \emph{TopDown principal components} are the outputs of the TopDown algorithm with input $(Y,k)$. The first $k$ \emph{BottomUp principal components} are the outputs of the BottomUp algorithm with input $(Y,k)$. 
\end{defn}

In TopDown, one first finds $E_k^\ast$. Then for $j = 1, 2, \ldots, k-1$, one finds $E_j$, the best $j$-dimensional subspace approximation to $Y-m_k^\ast$, subjected to $E_{j-1} \subset E_j \subset E_k^\ast$. This defines a nested sequence of subspace $E_1\subset E_2 \subset \ldots \subset E_{k-1} \subset E_k^\ast$, and hence a basis for $E_k^\ast$, such that $E_j$ is an approximation of the best $j$-dimensional subspace approximation to $Y-m_k^\ast$ contained in $E_k^\ast$. We solve (\ref{eqn:optim.estar.tau}) 
since $(m_k^\ast,E_k^\ast)$ is the true minimizer in dimension $k$, and thus we knew the optimal constant term. 

In BottomUp, one first finds $E_1^\ast$. Then for $j = 2, \ldots, k$, one finds $(m_j,E_j)$, the optimal $j$-dimensional affine subspace approximation to $Y$, subjected to $E_{j-1} \subset E_j$. In each step we re-estimate the constant term. Again, we obtain a nested sequence of subspaces $E_1^\ast \subset E_2 \subset \ldots \subset E_k$, and constant terms $m_1, \ldots, m_k$, where $(m_j,E_j)$ is an approximation to the best affine $j$-dimensional subspace approximation to $Y$. 

When $\tau = 1/2$, that is, when doing usual PCA, both algorithms correctly recover the principal components. For $\tau \neq 1/2$, they can produce different output. Interestingly, both in simulations and in practice, their outputs are not significantly different (see Sections~\ref{sec:sim} and \ref{sec:app}). See Section \ref{sec:alg} for a formal description of the TopDown and BottomUp algorithms and computational bounds on their convergence times.

\subsection{Statistical properties}
Even for $\tau = 1/2$, the objective function $J(U,V)$ is not simultaneously convex in both $U$ and $V$, but it is a convex function when either one of the two arguments is kept fixed. 
By the same argument, one can show that the same property holds for $J(U,V,W)$. That is, if $U$ is kept fixed, then $J(U,V,W)$ (which is now a function of V only, as W is a function of U and V) is convex in $V$. Similarly, if $V$ is kept fixed, then $J(U,V,W)$ is a convex function in $U$. Applying the result of \citet{hjort2011asymptotics}, we see that in each iteration, $V_n^{(t+1)}$ differs from $V^{(t+1)}$ by a term of order $\mathcal O(n^{-1/2})$. Thus, if the total number of iterations is small, one can prove consistency of the iterative least squares algorithm. We are not able to obtain a theoretical bound on the total number of iterations. In practice this does indeed seem to be small.

\section{Principal components as maximizers of captured variance} \label{sec:pec}

\subsection{Review of PCA}
Again, suppose we observe $n$ vectors $Y_1, \ldots, Y_n \in \R^p$. The first principal component $\phi^\ast$ is the unit vector in $\R^p$ which maximizes the variance of the data projected onto the subspace spanned by $\phi^\ast$. That is,
\begin{equation}\label{eqn:optim.variance}
\phi^\ast = \argmax_{\phi \in \R^p, \phi^\top\phi = 1}\textup{Var}(\phi\phi^\top Y_i: 1 \leq i \leq n) = \argmax_{\phi \in \R^p, \phi^\top\phi = 1}{n^{-1}}\sum_{i=1}^n (\phi^\top Y_i-\overline{\phi^\top Y})^2,
\end{equation}
where $\overline{\phi^\top Y} = n^{-1}\sum_{i=1}^n \phi^\top Y_i = \phi^\top \bar{Y}$ is the mean of the projected data, or equivalently, the projection of the mean $\bar{Y}$ onto the subspace spanned by $\phi$. Given that the first principal component is $\phi^\ast_1$, the second principal component $\phi^\ast_2$ is the unit vector in $\R^p$ which maximizes the variance of the residual $Y_i-(\phi_1^\ast)^\top\bar{Y}-\phi_1^\ast(\phi_1^\ast)^\top Y_i$, and so on. In this formulation, the data does not have to be pre-centered. The sum $(\phi_1^\ast)^\top\bar{Y} + (\phi_2^\ast)^\top\bar{Y} + \ldots + (\phi_{k}^\ast)^\top\bar{Y}$ is the overall mean $\bar{Y}$ projected onto the subspace spanned by the first $k$ principal components. 
For the benefit of comparison to Theorem \ref{thm:pec}, let us reformulate PCA as an optimization problem. Define 
\begin{equation}\label{eqn:C.pca}
C = {n^{-1}}\sum_{i=1}^n(Y_i - \bar{Y})(Y_i - \bar{Y})^\top.
\end{equation}
Then $\phi^\ast$ is the solution to the following optimization problem.
\begin{align*}
\mbox{maximize } & \phi^\top C\phi \\
\mbox{subject to } & \phi^\top\phi = 1.
\end{align*}
The principal component is not necessarily unique: if the covariance matrix is the identity, for example, then any unit vector $\phi$ would solve (\ref{eqn:optim.variance}), and thus there is no unique principal component. In the discussions that follows, we implicitly assume that the principal component $\phi^\ast$ is unique. In other words, $C$ has a unique largest eigenvalue. 

\subsection{An analogue for expectiles}
Let $Y \in \R$ be a random variable with cdf $F$. We define its \emph{$\tau$-variance} to be
$$ \textup{Var}_{\tau}(Y) = \rE\|Y - e_\tau\|_{\tau, 2}^2 = \min_{e \in \R}\rE\|Y - e\|_{\tau, 2}^2 $$
where $e_\tau = e_\tau(Y)$ is the $\tau$-expectile of $Y$. When $\tau = 1/2$, this reduces to the usual definition of variance. The following are immediate from Proposition \ref{prop:exp.prop}

\begin{prop}[Properties of $\tau$-variance] \label{prop:var.tau}
Let $Y \in \R$ be a random variable. For $\tau \in (0,1)$, the following statements hold.
\begin{itemize}
	\item $\textup{Var}_\tau(Y + c) = \textup{Var}_\tau(Y)$ for $c \in \R$
	\item $\textup{Var}_\tau(sY) = s^2\textup{Var}_\tau(Y)$ for $s \in \R, s > 0$. 
	\item $\textup{Var}_\tau(-Y) = \textup{Var}_{1-\tau}(Y)$
\end{itemize}
\end{prop}
\begin{proof} The first two follow directly from corresponding properties for $e_\tau$. We shall prove that last assertion. Recall that $e_\tau(-Y) = -e_{1-\tau}(Y)$. Thus
\begin{align*} \textup{Var}_\tau(-Y) &= \rE\|-Y - e_\tau(-Y)\|_{\tau,2}^2 = \rE\|-\{Y - e_{1-\tau}(Y)\}\|_{\tau,2}^2 = \rE\|Y - e_{1-\tau}(Y)\|_{1-\tau,2}^2 \\
&= \textup{Var}_{1-\tau}(Y). \end{align*}
\end{proof}

If $\phi \in \R^p$ is a unit vector, that is, $\phi^\top\phi =1$, then we define
$$ \textup{Var}_\tau(\phi\phi^\top Y_i: 1 \leq i \leq n) = \textup{Var}_\tau(\phi^\top Y_i: 1 \leq i \leq n). $$
That is, the $\tau$-variance of $n$ vectors which are multiples of $\phi$ is just the $\tau$-variance of the coefficients, which is a sequence of \emph{real numbers}. Thus, the direct generalization of (\ref{eqn:optim.variance}) would be 
\begin{align}
\phi^\ast_\tau & = \argmax_{\phi \in \R^p, \phi^\top \phi = 1}\textup{Var}_\tau(\phi\phi^\top  Y_i: 1 \leq i \leq n) = \argmax_{\phi \in \R^p, \phi^\top \phi = 1}\textup{Var}_\tau(\phi^\top Y_i: 1 \leq i \leq n) \label{eqn:optim.var.tau.raw} \\
&= \argmax_{\phi \in \R^p, \phi^\top \phi = 1}{n^{-1}}\sum_{i=1}^n(\phi^\top Y_i - \mu_\tau)^2w_i
\label{eqn:optim.variance.tau} 
\end{align}
where $\mu_\tau \in \R$ is the $\tau$-expectile of the sequence of $n$ real numbers $\phi^\top Y_1, \ldots \phi^\top Y_n$, and 
\begin{equation}\label{eqn:wi.variance}
w_i = \tau \mbox{ if } \sum_{j=1}^p Y_{ij}\phi_j > \mu_\tau, \mbox{ and } w_i = 1-\tau \mbox{ otherwise.}
\end{equation}

\begin{defn} Suppose we observe $Y_1, \ldots, Y_n \in \R^p$. The first \emph{principal expectile component} (PEC) $\phi_\tau^\ast$ is the unit vector in $\R^p$ that maximizes the $\tau$-variance of the data projected on the subspace spanned by $\phi_\tau^\ast$. That is, $\phi_\tau^\ast$ solves (\ref{eqn:optim.variance.tau}).
\end{defn}

`The' principal expectile component is not necessarily unique. In classical PCA, the first principal component is only unique if and only if the covariance matrix has a unique maximal eigenvalue. Even then, under this assumption, the principal component is only unique up to sign. That is, if $\phi$ is the principal component, then $-\phi$ is also a principal component. Principal expectile component, on the other hand, are sign-sensitive in general, unless if the distribution of $Y$ is symmetric, or if $\tau = 1/2$. We make this observation concrete below, which is a Corollary of Proposition \ref{prop:var.tau}. 
\begin{cor}
For $\tau \in (0,1)$, random variable $Y \in \R^p$, suppose $\phi_\tau^\ast$ is a first $\tau$-PEC of $Y$. Then
$$ -\phi_\tau^\ast = \phi_{1-\tau}^\ast,$$
that is, $-\phi_\tau^\ast$ is also a first $(1-\tau)$-PEC of $Y$. Furthermore, if the distribution of $Y$ is symmetric about $0$, that is, $Y \stackrel{\mathcal L}{=} -Y$, then
$-\phi_{\tau}^\ast$ is also a first $\tau$-PEC of $Y$. 
\end{cor}
\begin{proof}
By Proposition \ref{prop:var.tau}, $\textup{Var}_\tau(\phi_\tau^{\ast\top} Y) = \textup{Var}_{1-\tau}\{(-\phi_\tau^{\ast\top}) Y\}$. Thus if $\phi_\tau^\ast$ solves (\ref{eqn:optim.var.tau.raw}) for $\tau$, then $(-\phi_\tau)^\ast$ solves (\ref{eqn:optim.var.tau.raw}) for $1-\tau$. If the distribution of $Y$ is symmetric about $0$, then 
$$\textup{Var}_\tau(\phi_\tau^{\ast\top} Y) = \textup{Var}_{1-\tau}\{\phi_\tau^{\ast\top} (-Y)\} = \textup{Var}_\tau(\phi_\tau^{\ast\top} Y).$$ 
In this case $ -\phi_\tau^\ast = \phi_{1-\tau}^\ast$ is another $\tau$-PEC of $Y$.
\end{proof}

Like in classical PCA, the other components are defined based on the residuals, and thus by definition, they are orthogonal to the previously found components. Therefore one obtains a nested sequence of subspace which captures the tail variations of the data. 

By replacing the $\|\cdot\|_{\tau,2}^2$ norm with the $\|\cdot\|_{\tau,1}$ norm, one can define the analogue of principal component for quantiles. The analogue of $\tau$-variance is the $\tau$-deviation
$$ \textup{Dev}_\tau(Y) = \rE\|Y - q_\tau(Y)\|_{\tau,1} = \min_{q\in\R^p}\rE\|Y - q\|_{\tau,1}. $$
The $\tau$-deviation is linear rather than quadratic with respect to constants, that is, $\textup{Dev}_\tau(cY)=c\textup{Dev}_\tau(Y)$ for $c > 0$, we consider vectors in the $L_1$ unit ball rather than the $L_2$ unit ball. Define the $\tau$-deviance of $n$ vectors which are multiples of a vector $\psi \in \R^p$ to be the $\tau$-deviance of the coefficients. That is,
$$\textup{Dev}_\tau(\psi\psi^\top Y_i: 1 \leq i \leq n) = \textup{Dev}_\tau(\psi^\top Y_i: 1 \leq i \leq n) $$
This leads to the optimization problem
$$ \psi_\tau^\ast = \argmax_{\psi \in \R^p: \sum_j|\psi_j| = 1}\textup{Dev}_\tau(\psi\psi^\top Y_i: 1 \leq i \leq n). $$
\begin{defn} The first \emph{principal quantile component} $\psi_\tau^\ast$ is the $L_1$-unit vector in $\R^p$ that maximizes the $\tau$-deviation captured by the data projected on the subspace spanned by~$\psi_\tau^\ast$.
\end{defn}

Generalizing principal components to quantiles via its interpretation as variance maximizer is not new. \citet{fraimanlopez} define the first principal quantile direction $\psi$ to be the one that maximizes the $L_2$ norm of the $\tau$-quantile of the centered data, projected in the direction $\psi$. That is, $\psi$ is the solution of
$$ \max_{\psi \in \R^p: \psi^\top \psi = 1}\|\psi^\top q_\tau(Y - \rE Y)\|_{1/2,2}. $$
Their definition works for random variables in arbitrary Hilbert spaces. \citet{kongmizera} proposed the same definition but without centering $Y$ at $\rE Y$. These authors used the principal directions computed to study quantile level sets of distributions in small dimensions. Compared to these work, our definition is very natural, can be extended to Hilbert spaces, and in the case of expectile, satisfies many `nice' properties, some of which are shared by the principal directions of \citet{fraimanlopez}. For example, the PEC coincides with the classical PC when the distribution of $Y$ is elliptically symmetric.

\begin{prop}\label{proppec}[Properties of principal expectile component]
Let $Y \in \R^p$ be a random variable, $\phi_\tau^\ast(Y)$ its unique first principal expectile component.
\begin{enumerate}
	\item For any constant $c \in \R^p$, $\phi_\tau^\ast(Y+c) = \phi_\tau^\ast(Y)$. In words, the PEC is invariant under translations of the data.
	\item If $B \in \R^{p \times p}$ is an orthogonal matrix, then $\phi_\tau^\ast(BY) = B\phi_\tau^\ast(Y)$. In words, the PEC respects change of basis.
	\item If the distribution of $Y$ is elliptically symmetric about some point $c \in \R^p$, that is, there exists an invertible $p \times p$ real matrix $A$ such that $BA^{-1}(Y -c) \stackrel{\mathcal L}{=} A^{-1}(Y-c)$ for all orthogonal matrix $B$, then $\phi_\tau^\ast(Y) = \phi_{1/2}^\ast(Y)$. In this case, the PEC coincides with the classical PC regardless of $\tau$.
	\item If the distribution of $Y$ is spherically symmetric about some point $c \in \R^p$, that is, $B(Y -c) \stackrel{\mathcal L}{=} Y-c$ for all orthogonal matrix $B$, then all directions are principal. 
\end{enumerate}
\end{prop}
\begin{proof}
By the first part of Proposition \ref{prop:var.tau}:
\begin{align*} \textup{Var}_\tau\{\phi^\top (Y_i + c): i = 1, \ldots, n\} &= \textup{Var}_\tau(\phi^\top Y_i + \phi^\top c: i = 1, \ldots, n)\\
 &= \textup{Var}_\tau(\phi^\top Y_i: i = 1, \ldots, n). 
\end{align*}
This proves the first statement. For the second, note that
$$ \textup{Var}_\tau(\phi^\top BY_i: i = 1, \ldots, n) = \textup{Var}_\tau\{(B^\top\phi)^\top Y_i: i = 1, \ldots, n\}. $$
Thus if $\phi_\tau^\ast$ is the first $\tau$-PEC of $Y$, then $(B^\top)^{-1}\phi_\tau^\ast$ is the first $\tau$-PEC of $BY$. But $B$ is orthogonal, that is, $(B^\top)^{-1} = B$. hence $B\phi_\tau^\ast$ is the $\tau$-PEC of $BY$. This proves the second statement. For the third statement, by statement 1, we can assume $c \equiv 0$. Thus $Y = AZ$ where $BZ \stackrel{\mathcal L}{=} Z$ for all orthogonal matrices $B$. Write $A$ in its singular value decomposition $A = UDV$, where $D$ is a diagonal matrix with positive values $D_{ii} = d_i$ for $i = 1, \ldots p$, and $U$ and $V$ are $p \times p$ orthogonal matrices. Choosing $B = V^{-1}$ gives
$$ \phi_\tau^\ast(Y) = \phi_\tau^\ast(UDZ) = U\phi_\tau^\ast(DZ). $$
Now, by Proposition \ref{prop:var.tau}, since $d_j \geq 0$ for all $j$,
$$ \textup{Var}_\tau(\phi^\top DZ) = \textup{Var}_\tau(\sum_{j=1}^pd_jZ_j\phi_j) = \sum_j\phi_j^2d_j^2\textup{Var}_\tau(Z_j).$$
Since $\sum_j\phi_j^2 = 1$, $\textup{Var}_\tau(\phi^\top DZ)$ lies in the convex hull of the $p$ numbers $d_j^2\textup{Var}_\tau(Z_j)$ for $j = 1, \ldots p$. Therefore, it is maximized by setting $\phi$ to be the unit vector along the axis $j$ with maximal  $d_j^2\textup{Var}_\tau(Z_j)$.  But $Z \stackrel{\mathcal L}{=} BZ$ for all orthogonal matrices $B$, thus $Z_j \stackrel{\mathcal L}{=} Z_k$, hence $\textup{Var}_\tau(Z_j) = \textup{Var}_\tau(Z_k)$ for all indices $j,k = 1, \ldots, p$. Thus $\textup{Var}_\tau(\phi^\top DZ)$ is maximized when $\phi$ is the unit vector along the axis $j$ with maximal $d_j$. This is precisely the axis with maximal singular value of $A$, and hence is also the direction of the (classical) principal component of $DZ$. This proves the claim. The last statement follows immediately from the third statement. 
\end{proof}

To compute the principal expectile component $\phi_\tau^\ast$, one needs to optimize the right-hand side of (\ref{eqn:optim.variance.tau}) over all unit vectors $\phi$. Although this is a differentiable function in $\phi$, optimizing it is a difficult problem, since $\mu_\tau$ also depends on $\phi$, and does not have a closed form solution. However, in certain situations, for given weights $w_i$, not only $\mu_\tau$ but also $\phi_\tau^\ast$ \emph{have} closed form solutions. 

\begin{thm}\label{thm:pec} Consider (\ref{eqn:optim.variance.tau}). Suppose we are given the true weights $w_i$, which are either $\tau$ or $1-\tau$. Let $\tau_+ = \{i \in \{1, \ldots, n\}: w_i = \tau\}$ denote the set of observations $Y_i$ with `positive' labels, and $\tau_- = \{i \in \{1, \ldots, n\}: w_i = 1-\tau\}$ denote its complement. Let $n_+$ and $n_-$ be the sizes of the respective sets. Define an estimator $\hat{e}_\tau \in \R^p$ of the $\tau$-expectile via
\begin{equation}
\hat{e}_\tau = \frac{\tau\sum_{i \in \tau_+}Y_i + (1-\tau)\sum_{i \in \tau_-}Y_i}{ \tau n_+ + (1-\tau)n_-}.\label{eqn:mu.tau.linear}
\end{equation}
Define 
\begin{equation}\label{eqn:C}
C_\tau = \frac{\tau}{n}\left\{\sum_{i\in \tau_+}(Y_i-\hat{e}_\tau)(Y_i-\hat{e}_\tau)^\top \right\} + \frac{1-\tau}{n}\left\{\sum_{i\in \tau_-}(Y_i-\hat{e}_\tau)(Y_i-\hat{e}_\tau)^\top \right\}.\end{equation}
Then $\phi^\ast_\tau$ is the solution to the following optimization problem.
\begin{align}
\textup{maximize } & \phi^\top C_\tau \phi \notag \\
\textup{subject to } & \phi^\top Y_i > \phi^\top\hat{e}_\tau \Leftrightarrow i \in \tau_+ \label{eqn:constraint} \\
& \phi^\top\phi = 1. \notag
\end{align}
In particular, the PEC is the constrained classical PC of a weighted version of the covariance matrix of the data, centered at a constant possibly different from the mean. 
\end{thm}
\begin{proof}
Since the weights are the true weights coming from the true principal expectile component $\phi^\ast_\tau$, clearly $\phi^\ast_\tau$ satisfies the constraint in (\ref{eqn:constraint}). Now suppose $\phi$ is another vector in this constraint set. Then $\phi^\top\hat{e}_\tau$ is exactly $\mu_\tau$, the $\tau$-expectile of the sequence of $n$ real numbers $\phi^\top Y_1, \ldots, \phi^\top Y_n$. Therefore, the quantity we need to maximize in (\ref{eqn:optim.variance.tau}) reads
\begin{align*}
\frac{1}{n}\sum_{i=1}^n(\phi^\top Y_i - \mu_\tau)^2 w_i
&= \frac{\tau}{n}\sum_{i \in \tau_+}(\phi^\top Y_i - \phi^\top \hat{e}_\tau)^2 + \frac{1-\tau}{n}\sum_{i \in \tau_-}(\phi^\top Y_i - \phi^\top \hat{e}_\tau)^2 \\
&= \frac{\tau}{n}\sum_{i \in \tau_+}\phi^\top(Y_i - \hat{e}_\tau)(Y_i-\hat{e}_\tau)^\top \phi + \frac{1-\tau}{n}\sum_{i \in \tau_-}\phi^\top(Y_i - \hat{e}_\tau)(Y_i-\hat{e}_\tau)^\top \phi \\
&= \phi^\top C_\tau \phi.
\end{align*}
Thus the optimization problem above is indeed an equivalent formulation of 
(\ref{eqn:optim.variance.tau}), which was used to define $\phi^\ast_\tau$. Finally, the last observation follows by comparing the above with the optimization formulation for PCA, see the paragraph after (\ref{eqn:C.pca}). Indeed, when $\tau = 1/2$, $\hat{e}_{1/2} = \bar{Y}$, $C_{1/2} = C$, and we recover the classical PCA.
\end{proof}

Since $\hat{e}_\tau$ is a linear function in the $Y_i$, (\ref{eqn:constraint}) defines a system of linear constraints in the entries of $Y_i$ and $\phi_\tau^\ast$. Thus for each fixed sign sets $(\tau_+, \tau_-)$, there exist (not necessarily unique) local optima $\phi_\tau^\ast(\tau_+, \tau_-)$. There are $2^n$ possible sign sets, one of which corresponds to the global optima $\phi_\tau^\ast$ that we need. It is clear that finding the global optimum $\phi_\tau^\ast$ by enumerating all possible sign sets is intractable. However, in many situations, the constraint in (\ref{eqn:constraint}) is inactive. That is, the largest eigenvector of $C_\tau$ satisfies (\ref{eqn:constraint}) for free. In such situations, we call $\phi^\ast_\tau$ a \emph{stable solution}. Just like classical PCA, stable solutions are unique for matrices $C_\tau$ with unique principal eigenvalue. More importantly, we have an efficient algorithm for finding stable solutions, if they exist. 

\begin{defn}
For some given sets of weights $w = (w_i)$, define $e_\tau(w)$ via (\ref{eqn:mu.tau.linear}), $C_\tau(w)$ via~(\ref{eqn:C}). Let $\phi_\tau(w)$ be the largest eigenvector of $C_\tau(w)$. If $\phi_\tau(w)$ satisfies (\ref{eqn:constraint}), we say that $\phi_\tau(w)$ is a locally \emph{stable solution} with weight $w$.
\end{defn}

To find locally stable solutions, one can solve (\ref{eqn:optim.variance}) using iterative reweighted least squares: first initialize the $w_i$'s, compute estimators $\mu_\tau(w)$ and $\phi_\tau(w)$ ignoring the constraint (\ref{eqn:constraint}), update the weights via (\ref{eqn:wi.variance}), and iterates. At each step of this algorithm, one finds the principal component of a weighted covariance matrix with some approximate weight. Since there are only finitely many possible weight sets, the algorithm is guaranteed to converge to a locally stable solution if it exists. In particular, if the true solution to (\ref{eqn:optim.variance}) is stable, then for appropriate initial weights, the algorithm will find this value.  We call this algorithm PrincipalExpectile. We give a formal description of this algorithm in Section~\ref{sec:alg}.

\subsection{Statistical properties}
We now prove consistency of local maximizers of (\ref{eqn:optim.variance}). The main theorem in this section is the following.

\begin{thm}\label{thm:consistent}
Fix $\tau > 0$. Let $Y_n$ be the empirical version of $Y$, a random variable in $\R^p$ with finite second moment, distribution function $F$. Suppose $\phi^\ast = \phi^\ast_\tau$ is a unique global solution to (\ref{eqn:optim.variance}) corresponding to $Y$. Then for sufficiently large $n$, for any sequence of global solutions $\phi_n^\ast$ corresponding to $Y_n$, we have
$$ \phi_n^\ast \stackrel{F-a.s.}{\longrightarrow}\phi^\ast $$
in $\R^p$ as $n \to \infty$. 
\end{thm}
For the proof, we first need the following lemma.
\begin{lem}\label{lem:consistent} Let $Y_n$ be the empirical version of $Y$, a random variable in $\R^p$ with finite second moment and distribution function $F$. Then uniformly over all $\phi \in \R^p$ with $\phi^\top \phi = 1$, and uniformly over all $\tau \in (0,1)$,
$$ \textup{Var}_\tau(Y_n^\top \phi) \stackrel{F-a.s.}{\longrightarrow} \textup{Var}_\tau(Y^\top \phi).$$
\end{lem}
\begin{proof}
Since $Y_n$ is the empirical version of $Y$ and the set of all unit vectors $\phi \in \R^p, \phi^\top \phi = 1$ is compact, by the Cramer-Wold theorem, $Y_n^\top \phi \stackrel{\mathcal L}{\to} Y^\top \phi$ uniformly over all such unit vectors $\phi \in \R^p$. It then follows that $e_\tau$ and $\textup{Var}_\tau$, which are completely determined by the distribution function, also converge $F-a.s.$ uniformly over all $\phi$.   
\end{proof}
\begin{proof}[Proof of Theorem \ref{thm:consistent}]
Let $\mathbb{S}^{p-1}$ denote the unit sphere in $\R^p$. Equip $\R^p$ with the Euclidean norm $\|\cdot\|$. Define the map $V_Y: \mathbb{S}^{p-1} \to \R$, $V_Y(\phi) =  \textup{Var}_\tau(Y^\top \phi)$. Fix $\epsilon > 0$. We shall prove that there exists a $\delta > 0$ such that the global minimum of $V_{Y_n}$ is necessarily within $\delta$-distance of $\phi^\ast$.

Since $V_Y$ is continuous, $\mathbb{S}^{p-1}$ is compact, and $\phi^\ast$ is unique, there exists a sufficiently small $\delta > 0$ such that
$$ |V_Y(\phi) - V_Y(\phi^\ast)| < \epsilon \Rightarrow \|\phi - \phi^\ast\| < \delta $$
for $\phi \in \mathbb{S}^{p-1}$. In particular, if $\|\phi - \phi^\ast\| > \delta$, then
$$ V_Y(\phi^\ast) + \epsilon < V_Y(\phi). $$
By Lemma \ref{lem:consistent}, $V_{Y_n} \to V_Y$ as $n \to \infty$ uniformly over $\mathbb{S}^{p-1}$. In particular, there exists a large $N$ such that for all $n > N$, 
$$ |V_{Y_n}(\phi) - V_Y(\phi)| < \epsilon/6$$
for all $\phi \in \mathbb{S}^{p-1}$. Thus for $\phi \in \mathbb{S}^{p-1}$ such that $\|\phi - \phi^\ast\| > \delta$,
$$ V_{Y_n}(\phi) - V_Y(\phi^\ast) > \epsilon - \epsilon/6 = 5\epsilon/6. $$
Meanwhile, since $V_Y$ is continuous, one can choose $\epsilon' = \epsilon/6$, and thus obtain $\delta'$ such that
$$ |V_Y(\phi) - V_Y(\phi^\ast)| < \epsilon/6 \Leftarrow \|\phi - \phi^\ast\| < \delta'. $$
Then, for $\phi$ such that $\|\phi - \phi^\ast\| < \delta'$, 
$$V_{Y_n}(\phi) - V_Y(\phi^\ast) \leq |V_{Y_n}(\phi) - V_Y(\phi)| + |V_Y(\phi) - V_Y(\phi^\ast)| < \epsilon/6 + \epsilon/6 = \epsilon/3. $$
So far we have shown that if $\|\phi - \phi^\ast\| > \delta$, then $V_{Y_n}(\phi)$ is at least $5\epsilon/6$ bigger than $V_Y(\phi^\ast)$. Meanwhile, if $\|\phi - \phi^\ast\| < \delta'$, then $V_{Y_n}(\phi)$ is at most $\epsilon/3$ bigger than $V_Y(\phi^\ast)$. Thus the global minimum $\phi_n^\ast$ of $V_{Y_n}$ necessarily satisfy $\|\phi_n^\ast - \phi^\ast\| < \delta$. This completes the proof.
\end{proof}

\section{Algorithms}\label{sec:alg}

\subsection{TopDown and BottomUp}
We now describe how iterative weighted least squares can be adapted to implement TopDown and BottomUp. We start with a description of the asymmetric weighted least squares (LAWS) algorithm of Newey and Powell \citet{np}. The basic algorithm outputs a subspace without the affine term, and needs to be adapted. See \citet{gzhh} for a variation with smoothing penalty and spline basis. 

\begin{algorithm}[ht]
\caption{Asymmetric weighted least squares (LAWS)}\label{alg:laws}
\begin{algorithmic}[1]
\State Input: data $Y \in \R^{n \times p}$, positive integer $k < p$
\State Output: $\hat{E}^\ast_k$, an estimator of $E^\ast_k$, expressed in product form $\hat{E}^\ast_k = \hat{U}\hat{V}^\top $, where $\hat{U} \in \R^{n \times k}, \hat{V} \in \R^{p \times k}$.$\hat{U}, \hat{V}$ are unique up to multiplication by an invertible matrix.  
\Procedure{LAWS}{$Y,k$}
\State Set $V^{(0)}$ to be some rank-$k$ $p \times k$ matrix. 
\State Set $W^{(0)} \in \R^{n \times p}$ to be 1/2 everywhere.
\State Set $t = 0$.
\Repeat
	\State Update $U$: Set $U^{(t+1)} = \argmin_{U \in \R^{n \times k}}J(U,V^{(t)}, W^{(t)})$. \label{ln:up.u}
	\State \label{ln:up.uw} Update $W$: Set $W^{(t+1)}_{ij} = \tau$ if $Y_{ij} - \sum_l U^{(t+1)}_{il}V^{(t)}_{lk} > 0$, $W^{(t+1)}_{ij} = 1-\tau$ otherwise. 
	\State Update $V$: Set $V^{(t+1)} = \argmin_{V \in \R^{k \times p}}J(U^{(t+1)},V, W^{(t+1)})$.\label{ln:up.v}
	\State Update $W$: Set $W^{(t+1)}_{ij} = \tau$ if $Y_{ij} - \sum_l U^{(t+1)}_{il}V^{(t+1)}_{lk} > 0$, $W^{(t+1)}_{ij} = 1-\tau$ otherwise.\label{ln:up.vw}	
	\State Set t = t + 1
 \Until {$U^{(t+1)} = U^{(t)}, V^{(t+1)} = V^{(t)}, W^{(t+1)} = W^{(t)}$}. \\
 \Return $\hat{E}_k = U^{(t)}(V^{(t)})^\top $.
\EndProcedure
\end{algorithmic}
\end{algorithm}
  
\begin{prop}\label{prop:laws.gradient.descent} The LAWS algorithm is well-defined, and is a gradient descent algorithm. Thus it converges to a critical point of the optimization problem $(\ref{eqn:optim.estar})$. 
\end{prop}  
\begin{proof} First, we note that the steps in the algorithm are well-defined. For fixed $W$ and $V$, $J(U,V,W)$ is a quadratic in the entries of $U$. Thus the global minimum on line 8 has an explicit solution, see \citet{srebro, gzhh}. A similar statement applies to line 9. 

As noted in Section \ref{sec:errmin}, $J(U,V, W)$ is not jointly convex in $U$ and $V$, but as a function in $U$ for fixed $V$, it is a convex, continuously differentiable, piecewise quadratic function. The statement holds for $J(U,V,W)$ as a function in $V$ for fixed $U$. Hence lines \ref{ln:up.u} and \ref{ln:up.uw} is one step in a Newton-Raphson algorithm on $J(U, V, W)$ for fixed $V$. Similarly, lines \ref{ln:up.v} and \ref{ln:up.vw} is one step in a Newton-Raphson algorithm on $J(U, V, W)$ for fixed $U$. Thus the algorithm is a coordinatewise gradient descent on a coordinatewise convex function, hence converges.  
\end{proof}

If some columns of $U$ or $V$ are pre-specified, one can run LAWS and not update these columns in lines \ref{ln:up.u} and \ref{ln:up.v}. Thus one can use LAWS to find the optimal affine subspace by writing $\1m^\top  + E = \tilde{U}\tilde{V}$ with the first column of $\tilde{U}$ constrained to be~$\1$. Similarly, we can use this technique to solve the constrained optimization problems: 

\begin{itemize}
\setlength{\itemsep}{1pt}
	\item \emph{Find a rank-$k$ approximation $E_k$ whose span contains a given subspace of dimension $r < k$} 
	\item Solution: Constrain the first $r$ columns of $V^{(0)}$ to be a basis of the given subspace. 
	\item \emph{Find a rank-$k$ approximation whose span lies within a given subspace of dimension $r > k$.}
	\item Solution: Let $B \in \R^{n \times r}$ be a basis of the given subspace. Then the optimization problem becomes
$$ \min_{U \in \R^{r \times k}, V \in \R^{p \times k}} \|Y - BUV^\top \|^2_{\tau, 2}. $$
One can then apply the LAWS algorithm with variables $U$ and $V$. 
	\item \emph{Find a rank-$k$ approximation whose span contains a given subspace of dimension $r < k$, and is contained in a given subspace of dimension $R > k$.}
	\item Solution: Combine the previous two solutions.
\end{itemize}

With these tools, we now define the two algorithms, TopDown and BottomUp.

\begin{algorithm}[ht]
\caption{TopDown}\label{alg:td}
\begin{algorithmic}[1]
\State Input: data $Y \in \R^{n \times p}$, positive integer $k < p$
\State Output: $\hat{E}^\ast_k$, an estimator of $E^\ast_k$, expressed in product form $\hat{E}^\ast_k = \hat{U}\hat{V}^\top $, where $\hat{U} \in \R^{n \times k}, \hat{V} \in \R^{p \times k}$ are unique. 
\Procedure{TopDown}{$Y,k$}
\State Use LAWS(Y,k) to find $\hat{m}^\ast_k, \hat{E}^\ast_k$. Write $\hat{E}^\ast_k = UV^\top $ for some orthonormal basis $U$. 
\State Use LAWS to find $\hat{U}_1$, the vector which spans the optimal subspace of dimension 1 contained in $U$. 
\State Use LAWS to find $\hat{U}_2$, where $(\hat{U}_1,\hat{U}_2)$ spans the optimal subspace of dimension 1 contained in $U$ and contains the span of $\hat{U}_1$
\State Repeat the above step until obtains $\hat{U}$.
\State Obtain $\hat{V}$ through the constraint $\hat{E}^\ast_k = \hat{U}\hat{V}^\top $. \\
\Return $\hat{m}^\ast_k, \hat{E}^\ast_k, \hat{U}, \hat{V}^\top $. 
\EndProcedure
\end{algorithmic}
\end{algorithm}

\begin{algorithm}[h]
\caption{BottomUp}\label{alg:bu}
\begin{algorithmic}[1]
\State Input: data $Y \in \R^{n \times p}$, positive integer $k < p$
\State Output: $\hat{E}^\ast_k$, an estimator of $E^\ast_k$, expressed in product form $\hat{E}^\ast_k = \hat{U}\hat{V}^\top $, where $\hat{U} \in \R^{n \times k}, \hat{V} \in \R^{p \times k}$ are unique. 
\Procedure{BottomUp}{$Y,k$}
\State Use LAWS to find $\hat{E}^\ast_1$. Let $\hat{U}_1$ be the basis vector. 
\State Use LAWS to find $\hat{U}_2$ such that $(\hat{U}_1, \hat{U}_2)$ is the best two-dimensional approximation to $Y$, subjected to containing $\hat{U}_1$.
\State Repeat the above step until obtains $\hat{U}$. We obtain $\hat{V}$ and $\hat{E}^\ast_k$ in the last iteration.
\Return $\hat{E}^\ast_k, \hat{U}, \hat{V}^\top $. 
\EndProcedure
\end{algorithmic}
\end{algorithm}

The TopDown algorithm requires the weights $w_{ij}$ and the loadings on previous principal components to be re-evaluated when finding the next principal component. A variant of the algorithm would be to keep the weights $w_{ij}$. In this case, the algorithm is still well-defined. However, it will produce a different basis matrix $\hat{U}$, since the estimators are no longer optimal in the $\|\cdot\|_{\tau,2}^2$ norm. 

\subsection{Performance bounds of TopDown and BottomUp}\label{ssec:bounds}
We now show that the dependence on $k$ only grows polylog in $n$. Thus both TopDown and BottomUp are fairly efficient algorithms even for large $k$.

\begin{thm}\label{thm:computation}
For fixed $V$ of dimension k, LAWS requires at most $\mathcal O\{\log(p)^k\}$ iterations, $\mathcal O\{npk^2\log(p)^k\}$ flops to estimate $U$.
\end{thm}

In other words, if $V$ has converged, LAWS needs at most $\mathcal O\{npk^2\log(p)^k\}$ flops to estimate $U$. The role of $U$ and $V$ are interchangeable if we transpose $Y$. Thus if $U$ has converged, LAWS needs at most $\mathcal O\{npk^2\log(n)^k\}$ to estimate $V$. We do not have a bound for the number of iterations needed until convergence. In practice this seem to be of order $\log$ of $n$ and $p$. For the proof of Theorem \ref{thm:computation} we need the following two lemmas.

\begin{lem}\label{lem:log.n} If $Y_1, \ldots, Y_n \in \R$ are $n$ real numbers, then LAWS finds their $\tau$-expectile $e_\tau$ in $\mathcal O\{\log(n)\}$ iterations.
\end{lem}
\begin{proof}
Given the weights $w_1, \ldots, w_n$, that is, given which $Y_i$'s are above and below $e_\tau$, the $\tau$-expectile $e_\tau$ is a linear function in the $Y_i$ as we saw in (\ref{eqn:mu.tau.linear}). As shown in Proposition \ref{prop:laws.gradient.descent}, LAWS is equivalent to a Newton-Raphson algorithm on a piecewise quadratic function. Since the points $Y_i$'s are ordered, it takes $\mathcal O\{\log(n)\}$ to learn their true weights. Thus the algorithm converges in $\mathcal O\{\log(n)\}$ iterations.
\end{proof}

\begin{lem}\label{lem:orthant} An affine line in $\R^p$ can intersect at most $2p$ orthants. 
\end{lem}
\begin{proof}
Recall that an \emph{orthant} of $\R^p$ is a subset of $\R^p$ where the sign of each coordinate is constrained to be either nonnegative or nonpositive. There are $2^p$ orthants in $\R^p$. Let $f(\lambda) = Y + \lambda v$ be our affine line, $\lambda \in \R, Y, v \in \R^p$. Let $\textup{sgn}: \R^p \to \{\pm 1\}^p$ denote the sign function. Now, $\textup{sgn}\{f(0)\} = \textup{sgn}(Y), \textup{sgn}\{f(\infty)\} = \textup{sgn}(v)$, and $\textup{sgn}\{f(\lambda)\}$ is a monotone increasing function in $\lambda$. As $\lambda \to \infty$, $\textup{sgn}\{f(\lambda)\}$ goes from $\textup{sgn}(Y)$ to $\textup{sgn}(v)$ one bit flip at a time. Thus there are at most $p$ flips, that is, the half-line $f(\lambda)$ for $\lambda \in [0, \infty)$ intersects at most $p$ orthants. By a similar argument, the half-line $f(\lambda)$ for $\lambda \in (-\infty, 0)$ intersects at most $p$ other orthants. This concludes the proof. 
\end{proof}

\begin{cor}\label{cor:orthant}An affine subspace of dimension $k$ in $\R^p$ can intersect at most $\mathcal O(p^k)$ orthants.
\end{cor}
\begin{proof}
Fix any basis, say $\psi_1, \ldots, \psi_k$. By Lemma \ref{lem:orthant}, $\psi_1$ can intersect at most $2p$ orthants. For each orthant of $\psi_1$, varying along $\psi_2$ can yield at most another $2p$ orthants. The proof follows by induction. (This is a rather liberal bound, but it is of the correct order for $k$ small relative to $p$).
\end{proof}

\begin{proof}[Proof of Theorem \ref{thm:computation}]
By Corollary \ref{cor:orthant}, it is sufficient to consider the case $k = 1$. Fix $V$ of dimension 1. Since $U, V$ are column matrices, we write them in lower case letters $u, v$. Solving for each $u_i$ is a separate problem, thus we have $n$ separate optimization problem, and it is sufficient to prove the claim for each $i$ for $i = 1, \ldots, n$. \\
Fix an $i$. As $u_i$ varies, $Y_i - m_i - u_iv$ defines a line in $\R^p$. The weight vector $(w_{i1}, \ldots, w_{ip})$ only depend on which coordinates are the orthant of $\R^p$ in which $Y_i -  m_i - u_iv$ is in. The later is equivalently to determining the weight of the $p$ points $\frac{Y_i-m_i}{v_i}$. By Lemma \ref{lem:log.n}, it takes $\mathcal O\{\log(p)\}$ for LAWS to determine the weights correctly. Thus LAWS takes at most $\mathcal O\{\log(p)\}$ iterations to converge, since each iteration involves estimating $w$, then $v$. Each iteration solves a weighted least squares, thus take $\mathcal O(npk^2)$. Hence for fixed $v$, LAWS can estimate $u$ after at most $\mathcal O\{npk^2\log(p)\}$ flops for $k = 1$. This concludes the proof for fixed $v$. 
By considering the transposed matrix $Y$, we see that the role of $u$ and $v$ are interchangeable. The conclusion follows similarly for fixed $u$. 
\end{proof}

\subsection{PrincipalExpectile}
In this section we describe the PrincipalExpectile algorithm. This algorithm is used to compute the principal expectile component defined in Section \ref{sec:pec}. We shall describe the case $k = 1$, that is, the algorithm for computing the first principal expectile component only. To obtain higher order components, one iterates the algorithm over the residuals $Y_i - \hat\phi_1(\hat\phi_1^\top Y_i + \hat{\mu}_1)$, where $\hat\mu_1$ is the $\tau$-expectile of the loadings $\hat\phi_1^\top Y_i$.

\begin{algorithm}
\caption{PrincipalExpectile}\label{alg:pec}
\begin{algorithmic}[1]
\State Input: data $Y \in \R^{n \times p}$.
\State Output: a vector $\hat{\phi}$, an estimator of the first principal expectile component of $Y$.
\Procedure{PrincipalExpectile}{$Y$}
\State Initialize the weights $w_i^{(0)}$
\State Set $t = 0$.
\Repeat
	\State Let $\tau^{(t)}_+$ be the set of indices $i$ such that $w_i^{(t)} = \tau$, and $\tau^{(t)}_-$ be the complement.
	\State Compute $e_\tau^{(t)}$ as in equation (\ref{eqn:mu.tau.linear}) with sets $\tau^{(t)}_+, \tau^{(t)}_-$.
	\State Compute $C_\tau^{(t)}$ as in equation (\ref{eqn:C}) with sets $\tau^{(t)}_+, \tau^{(t)}_-$.
	\State Set $\phi^{(t)}$ to be the largest eigenvector of $C_\tau^{t}(C_\tau^{t})^\top$
	\State Set $\mu_\tau^{(t)}$ to be the $\tau$-expectile of $(\phi^{(t)})^\top Y_i$
	\State Update $w_i$: set $w_i^{(t+1)} = \tau$ if $(\phi^{(t)})^\top Y_i > \mu_\tau^{(t)}$, and set $w_i^{(t+1)} = 1-\tau$ otherwise.
	\State Set t = t + 1
 \Until $w_i^{t} = w_i^{(t+1)}$ for all $i$. \\
 \Return $\hat\phi = \phi^{(t)}$.
\EndProcedure
\end{algorithmic}
\end{algorithm}

For $n$ observations $Y_1, \ldots, Y_n$, there are at most $2^n$ possible labels for the $Y_i$'s, and hence the algorithm has in total $2^n$ possible values for the $w_i$'s. Thus either Algorithm \ref{alg:pec} converges to a point which satisfies the properties of the optimal solution that Theorem \ref{thm:pec} prescribes, or that it iterates infinitely over a cycle of finitely many possible values of the $w_i$'s. In particular, the true solution is a fixed point, and thus fixed points always exist. In practice, we find that the algorithm converges very quickly, and can get stuck in a finite cycle of values. In this case, one can jump to a different starting point and restart the algorithm. Choosing a good starting value is important in ensuring convergence. Since the $\tau$-variance is a continuous function in $\tau$, we find that in most cases, one can choose a good starting point by performing a sequence of such computations for a sequence of $\tau$ starting with $\tau = 1/2$, and set the initial weight to be that induced by the previous run of the algorithm for a slightly smaller (or larger) $\tau$. 
\newpage\section{Simulation}\label{sec:sim}
\input{simulation_application}


\section{Summary}
\input{summary}

\bibliography{references}
\bibliographystyle{ecta}

\end{document}

%% file: preamble.tex
\usepackage[body={6.3in,9.3in}]{geometry}
\usepackage{verbatim}
\usepackage[ansinew]{inputenc}

\usepackage{textcomp}
\usepackage{latexsym}
\usepackage[usenames]{color}
\usepackage{amsmath,amsfonts,amssymb,amsthm}
\usepackage{graphicx}
\usepackage{longtable}
\usepackage{bbm}
\usepackage{algorithm, algorithmicx, algpseudocode}

\usepackage{setspace}
\usepackage{array}

\def\qed{\hfill{\raggedleft{\hbox{$\Box$}}} \smallskip}

\def\R{\mathbb{R}}

\DeclareMathOperator*{\argmax}{argmax}
\DeclareMathOperator*{\argmin}{argmin}

\newcommand{\rE}{\text{\sffamily{E}}}

\theoremstyle{plain} \newtheorem{lem}{Lemma}[section]
\theoremstyle{plain} \newtheorem{prop}{Proposition}[section]
\theoremstyle{plain} \newtheorem{thm}{Theorem}[section]
\theoremstyle{plain} \newtheorem{cor}{Corollary}[section]
\theoremstyle{plain} 
\theoremstyle{plain} 
\theoremstyle{definition} \newtheorem{defn}{Definition}[section]
\theoremstyle{definition}
\theoremstyle{definition} 
\theoremstyle{definition} 
\theoremstyle{definition} 
\theoremstyle{definition}

\newlength\savedwidth

%% file: simulation_application.tex
To study the finite sample properties of the proposed algorithms we do a simulation study. We follow the simulation setup of \cite{gzhh}, that is, we simulate the data $Y_{ij}, i=1,\ldots,n$, $j=1,\ldots,p$ as
\begin{equation}\label{simy}
Y_{ij}=\mu(t_j)+f_1(t_j)\alpha_{1i}+f_2(t_j)\alpha_{2i}+\varepsilon_{ij},
\end{equation}
where $t_j$'s are equidistant on [0,1], $\mu(t)=1+t+\exp\{-(t-0.6)^2/0.05\}$ is the mean function, $f_1(t)=\sqrt 2\sin(2\pi t)$ and $f_2(t)=\sqrt 2\cos(2\pi t)$ are principal component curves, and $\varepsilon_{ij}$ is a random noise.\\
We consider different settings 1 and 2 each with five error scenarios:
\begin{enumerate}
\item $\alpha_{1i}\sim \textup{N}(0,36)$ and $\alpha_{2i}\sim \textup{N}(0,9)$ are both iid and $\varepsilon_{ij}$'s are
(1) iid $\textup{N}(0, \sigma_1^2)$, (2) iid $t(5)$, (3) independent $\textup{N}\{0,\mu(t_j)\sigma_1^2\}$, (4) iid $\textup{logN}(0,\sigma_1^2)$ and (5) iid sums of two uniforms $U(0,\sigma_1^2)$ with $\sigma_1^2$=0.5.
\item $\alpha_{1i}\sim \textup{N}(0,16)$ and $\alpha_{2i}\sim \textup{N}(0,9)$ are both iid and $\varepsilon_{ij}$'s are 
(1) iid $\textup{N}(0, \sigma_2^2)$, (2) iid $t(5)$, (3) independent $\textup{N}\{0,\mu(t_j)\sigma_2^2\}$, (4) iid $\textup{logN}(0,\sigma_2^2)$ and (5) iid sums of two uniforms $U(0,\sigma_2^2)$ with $\sigma_2^2$=1.
\end{enumerate}
Note that the settings imply different ratios of coefficient-to-coefficient-to-noise variations. In the setting 1 scenario (1) we have a ratio 36:9:0.5, whereas in the setting 2 scenario (1) we have 16:9:1. Apart from standard Gaussian errors, we also consider "fat tailed" errors in scenario (2), heteroscedastic in (3) and skewed errors in (4). 
We study the performance of the algorithms for three sample sizes: (i) small $n$=20, $p$=100; (ii) medium $n$=50, $p$=150; (iii) large $n$=100, $p$=200.

For every combination of parameters we repeat the simulations 500 times and record the mean computing times, the mean of the average mean squared error (MSE), its standard deviation, and convergence ratio for each algorithm. We label the run of the algorithm as unconverged whenever after 30 iterations and 50 restarts from a random starting point the algorithms fail to converge. 

Computational time records are in Table \ref{tab4} and convergence statistics are reported in Table \ref{tab3}. For the ease of notation we write BUP for BottomUp, TD for TopDown and PEC for PrincipalExpectile.

PEC is the fastest algorithm as shown in Table \ref{tab4}. For large sample and high expectile level it is more than three times faster than TD and more than five times faster than BUP. Although the fastest from considered algorithms the PEC is considerably slower than the classical PCA routines with the average computational times 0.002 seconds for small, 0.005 seconds for medium, and 0.023 seconds for large sample (computed by function \verb prcomp ~in package \verb stats ~of statistical script language R).  
\begin{table}[H]
\centering
\begin{tabular}{cccccccccc}\hline\hline
sample &\multicolumn{3}{c}{small}&\multicolumn{3}{c}{medium}&\multicolumn{3}{c}{large}\\
$\tau$/sec&BUP&TD&PEC&BUP&TD&PEC&BUP&TD&PEC\\\hline
0.900    &         1.15   &       0.70      &   { 0.57}    &      2.87     &     1.59& { 1.39}      &    7.44   &       4.02     &      {2.71}\\
0.950      &      1.52     &    1.13       &   {  0.55}    &      3.94       &   2.68&{1.57}   &      10.34       &   6.88    &       {3.03}\\
0.975      &     2.47       &   2.32       &   { 0.56}     &     5.49       &   4.62&  {1.56}    &     14.37   &     10.96       &    {3.54}\\
\hline\hline	
\end{tabular}
\caption{Average time in seconds for convergence of the algorithms by 500 simulations}\label{tab4}
\end{table}

The major draw back of PEC is the relative low convergence rate: for all sample sizes only around 80\% of algorithm runs were convergent. In 20\% cases the algorithm keeps iterating between two sets of weights which possibly indicates an adverse sample geometry, i.e. that two eigevalues of the scaled covariance matrix are too close to each other. TD, on the contrary, converges almost always in medium and large sample sizes.
\begin{table}[H]
\centering
\begin{tabular}{cccccccccc}\hline\hline
sample &\multicolumn{3}{c}{small}&\multicolumn{3}{c}{medium}&\multicolumn{3}{c}{large}\\
$\tau$/rate&BUP&TD&PEC&BUP&TD&PEC&BUP&TD&PEC\\\hline
0.900    &         0.11      &   0.00       &   {0.24}        &  0.07      &   0.00   &  {0.23}      &    0.03        &    0.00     &    { 0.20}\\
0.950      &      0.17     &    0.00       &    {0.22}   &      0.13    &     0.00 &  {0.26}    &      0.11       &     0.00    &       {0.21}\\
0.975      &    0.25     &    0.03      &     {0.21}     &    0.22     &   0.01  &  {0.25}     &    0.22      &      0.00       &   {0.24}\\
\hline\hline	
\end{tabular}
\caption{Nonconvergence rates of the algorithms by 500 simulation runs}\label{tab3}
\end{table}

The results on the MSEs for both simulation settings are presented in Tables \ref{tab1} and \ref{tab2} respectively. For the settings 1 and 2 solely the magnitude of the average MSE differs; there is no substantial qualitative difference in relative performance of the algorithms. BUP performs the worst of the three algorithms in terms of its MSE in all scenarios. TD and PEC are comparable in terms of their MSEs. PEC shows robustness against skewness and fat tails in the error distribution since it produces the lowest MSEs in scenarios (2) and (4). Yet TD tends to slightly outperform PEC in medium and large samples by errors close to iid normal or normal heteroscedastic; by small sample sizes PEC outperforms TD in all scenarios but (5). 

Figures \ref{figpc1} and \ref{figpc2} illustrate the difference in the quality of component estimation for the 95\% expectile when coefficient-to-coefficient-to-noise variation ratio changes (setting 1 versus setting 2 respectively). The results are shown for the error scenario (1) and  small sample size. We observe that as the ratio changes from 36:9:0.5 (setting 1, Figure \ref{figpc1}) to 16:9:1 (setting 2, Figure \ref{figpc2}) the variability of the estimators of both component functions increases. The overall mean of the estimators remains very close to the true component functions.

We conclude that whenever the error distribution is fat-tailed or skewed, or by small samples PEC is likely produce more reliable results in terms of its MSE, whereas by errors close to normal and moderate or large samples TD is likely to produce smaller MSEs. 
\begin{spacing}{0.05}
\begin{landscape}
\begin{table}
\centering
\begin{small}
\begin{tabular}{rrrrrrrrrrr}\hline\hline
\multirow{2}{*}{scenario}&\multirow{2}{*}{$\tau$}&\multicolumn{3}{c}{$n$=20, $p$=100}&\multicolumn{3}{c}{$n$=50, $p$=150}&\multicolumn{3}{c}{$n$=100, $p$=200}\\
&&BUP& TD& PEC& BUP& TD& PEC& BUP& TD& PEC\\\hline
\multirow{6}{*}{(1)} &0.900&0.2762&0.1216&0.1123&0.1339&0.0538&0.0632 &0.0698&0.0297&0.0459\\
&&(0.1997) &(0.0097)&(0.0111)&(0.1099)&(0.0033)&(0.0029)&(0.0552)&(0.0015)&(0.0014)\\
&0.950&0.3619 &0.1568 &0.1334 &0.2323 &0.0705 &0.0727 &0.1312 &0.0394 &0.051\\
&&(0.2199) &(0.0123) &(0.0181) &(0.2076) &(0.0045) &(0.0044) &(0.1415) &(0.0020) &(0.0019)\\
&0.975&0.5064 &0.2053 &0.1601 &0.3583 &0.0944 &0.0874 &0.2157 &0.0536 &0.0594\\
&&(0.2977) &(0.0154) &(0.0276) &(0.2989) &(0.0060) &(0.0075) &(0.2314) &(0.0027) &(0.0035)\\
\multirow{6}{*}{(2)} &0.900&0.7092 &0.5421 &0.3147 &0.3382 &0.2714 &0.1494 &0.1866 &0.1548 &0.0932\\
&&(0.2382) &(0.1096) &(0.0685) &(0.1223) &(0.0727) &(0.0117) &(0.0522) &(0.0217) &(0.0050)\\
&0.950&1.105 &0.7847 &0.3854 &0.5789 &0.4440 &0.1819 &0.3316 &0.2680 &0.1101\\
&&(0.4453) &(0.1646) &(0.0988) &(0.2664) &(0.1675) &(0.0192) &(0.1144) &(0.0575) &(0.0075)\\
&0.975&1.6066 &1.1158 &0.4709 &0.9956 &0.7033 &0.2309 &0.5780 &0.4641 &0.1358\\
&&(0.7968) &(0.2106) &(0.1413) &(0.6936) &(0.2629) &(0.0341) &(0.2227) &(0.1175) &(0.0132)\\
\multirow{6}{*}{(3)} &0.900&0.4146 &0.2300 &0.2215 &0.1829 &0.1019 &0.1270 &0.0962 &0.0562 &0.0942\\
&&(0.2413) &(0.0195) &(0.0236) &(0.1070) &(0.0065) &(0.0066) &(0.0510) &(0.0029) &(0.0032)\\
&0.950&0.6261 &0.2966 &0.2792 &0.3538 &0.1335 &0.1622 &0.1603 &0.0746 &0.1208\\
&&(0.6313) &(0.0246) &(0.0369) &(1.1684) &(0.0088) &(0.0097) &(0.1135) &(0.0039) &(0.0045)\\
&0.975&0.8051 &0.3885 &0.3516 &0.4879 &0.1789 &0.2109 &0.2665 &0.1016 &0.1568\\
&&(0.4516) &(0.0312) &(0.0527) &(0.3736) &(0.0118) &(0.0167) &(0.2234) &(0.0052) &(0.0077)\\
\multirow{6}{*}{(4)} &0.900&0.9162 &0.8041 &0.2226 &0.4854 &0.4510 &0.1077 &0.2876 &0.2763 &0.0697\\
&&(0.2432) &(0.1532) &(0.0588) &(0.1093) &(0.0597) &(0.0089) &(0.0498) &(0.0247) &(0.0042)\\
&0.950&1.4972 &1.2869 &0.2725 &0.9127 &0.8092 &0.1296 &0.5585 &0.5280 &0.0812\\
&&(0.4494) &(0.2337) &(0.0713) &(0.4895) &(0.1187) &(0.0142) &(0.1595) &(0.0554) &(0.0069)\\
&0.975&2.3371 &1.9727 &0.3331 &1.5522 &1.3387 &0.1629 &1.2223 &0.9421 &0.0995\\
&&(1.0034) &(0.2835) &(0.0979) &(0.7483) &(0.1999) &(0.0248) &(1.4707) &(0.1110) &(0.0117)\\
\multirow{6}{*}{(5)}&0.900&0.0343 &0.0091 &0.0368 &0.0298 &0.0038 &0.0315 &0.0244 &0.0021 &0.0296\\
&&(0.0224) &(0.0007) &(0.0013) &(0.0261) &(0.0002) &(0.0004) &(0.0238) &(0.0001) &(0.0002)\\
&0.950&0.1225 &0.0110 &0.0409 &0.0351 &0.0044 &0.0345 &0.0285 &0.0023 &0.0322\\
&&(1.1145) &(0.0008) &(0.0020) &(0.0398) &(0.0003) &(0.0007) &0.0254 &(0.0004) &(0.0004)\\
&0.975&0.0776 &0.0135 &0.0474 &0.0455 &0.0052 &0.0397 &0.0360 &0.0027 &0.0366\\
&&(0.3266) &(0.0011) &(0.0034) &(0.0658) &(0.0003) &(0.0012) &(0.0309) &(0.0001) &(0.0006)\\\hline\hline
\end{tabular}
\caption{average MSE and its standard deviation in brackets by 500 simulation runs for the simulation setting 1.}
\label{tab1}
\end{small}
\end{table}
\begin{table}
\begin{small}
\centering
\begin{tabular}{rrrrrrrrrrr}\hline\hline
\multirow{2}{*}{scenario}&\multirow{2}{*}{$\tau$}&\multicolumn{3}{c}{$n$=20, $p$=100}&\multicolumn{3}{c}{$n$=50, $p$=150}&\multicolumn{3}{c}{$n$=100, $p$=200}\\
&&BUP& TD& PEC& BUP& TD& PEC& BUP& TD& PEC\\\hline
\multirow{6}{*}{(1)} &0.900&0.4484 &0.2436 &0.1988 &0.2053 &0.1077 &0.1002 &0.1109 &0.0595 &0.0660\\
&&(0.2671) &(0.0195) &(0.0238) &(0.1273) &(0.0066) &(0.0058) &(0.0924) &(0.0030) &(0.0027)\\
&0.950&0.7021 &0.314 &0.2418 &0.3681 &0.1411 &0.119 &0.2075 &0.0788 &0.0761\\
&&(0.4611) &(0.0246) &(0.0386) &(0.3066) &(0.0090) &(0.0091) &(0.2346) &(0.0039) &(0.0039)\\
&0.975&0.9218 &0.4116 &0.2945 &0.5957 &0.1890 &0.1483 &0.3364 &0.1074 &0.0925\\
&&(0.5578) &(0.0312) &(0.0546) &(0.4751) &(0.0121) &(0.0152) &(0.3565) &(0.0053) &(0.0067)\\
\multirow{6}{*}{(2)}&0.900&0.7424 &0.5427 &0.3186 &0.3560 &0.2716 &0.1502 &0.2047 &0.1549 &0.0935\\
&&(0.2933) &(0.1099) &(0.0762) &(0.1695) &(0.0728) &(0.0123) &(0.1886) &(0.0218) &(0.0050)\\
&0.950&1.1483 &0.7855 &0.3920 &0.6656 &0.4437 &0.1832 &0.3805 &0.2684 &0.1103\\
&&(0.5078) &(0.1643) &(0.1096) &(0.6719) &(0.1658) &(0.0185) &(0.3563) &(0.0581) &(0.0075)\\
&0.975&1.7083 &1.1095 &0.4805 &1.1714 &0.7048 &0.2342 &0.6974 &0.4648 &0.1368\\
&&(0.8614) &(0.1744) &(0.1493) &(0.9716) &(0.2652) &(0.0323) &(0.5981) &(0.1192) &(0.0126)\\
\multirow{6}{*}{(3)} &0.900&0.6616 &0.4613 &0.4093 &0.2993 &0.2041 &0.2200 &0.1684 &0.1126 &0.1540\\
&&(0.2625) &(0.0392) &(0.0486) &(0.1163) &(0.0131) &(0.0134) &(0.1880) &(0.0058) &(0.0066)\\
&0.950&1.0027 &0.5948 &0.5229 &0.4979 &0.2675 &0.2875 &0.3031 &0.1494 &0.2042\\
&&(0.5055) &(0.0495) &(0.0802) &(0.3671) &(0.0177) &(0.0215) &(0.4360) &(0.0077) &(0.0090)\\
&0.975&1.465 &0.7811 &0.6719 &0.8605 &0.3587 &0.3831 &0.5173 &0.2036 &0.2724\\
&&(0.8018) &(0.0627) &(0.1154) &(0.8004) &(0.0237) &(0.0338) &(0.6708) &(0.0103) &(0.0156)\\
\multirow{6}{*}{(4)}&0.900&5.4073 &5.2042 &1.0318 &3.3226 &3.2871 &0.4075 &2.0358 &2.0686 &0.2295\\
&&(2.1503) &(1.9812) &(0.9534) &(1.1548) &(1.0106) &(0.1258) &(0.6044) &(0.5259) &(0.1632)\\
&0.950&8.7171 &8.0696 &1.4256 &6.5227 &6.2094 &0.5143 &4.5541 &4.4481 &0.2939\\
&&(2.8223) &(2.3418) &(1.4550) &(1.9576) &(1.5846) &(0.1540) &(1.4193) &(1.0287) &(0.3150)\\
&0.975&13.419 &11.635 &2.0054 &11.202 &9.8804 &0.7372 &8.9280 &8.3663 &0.3889\\
&&(5.1223) &(1.6721) &(2.2733) &(4.0968) &(1.8550) &(0.5037) &(2.4679) &(2.7240) &(0.3161)\\
\multirow{6}{*}{(5)}&0.900&0.1135 &0.0365 &0.0572 &0.0923 &0.0153 &0.0394 &0.0561 &0.0083 &0.0333\\
&&(0.0755) &(0.0027) &(0.0041) &(0.0878) &(0.0009) &(0.0011) &(0.0628) &(0.0004) &(0.0005)\\
&0.950& 0.1430 &0.0440 &0.0651 &0.1197 &0.0177 &0.0434 &0.0896 &0.0093 &0.0356\\
&&(0.1214) &(0.0034) &(0.0060) &(0.1033) &(0.0010) &(0.0018) &(0.0938) &(0.0005) &(0.0008)\\
&0.975& 0.2489 &0.0540 &0.0769 &0.1538 &0.0209 &0.0499 &0.1145 &0.0107 &0.0396\\
&&(0.6091) &(0.0042) &(0.0099) &(0.1272) &(0.0013) &(0.0031) &(0.1042) &(0.0006) &(0.0013)\\\hline\hline
\end{tabular}
\caption{average MSE and its standard deviation in brackets by 500 simulation runs for the simulation setting 2.}
\label{tab2}
\end{small}
\end{table}
\end{landscape}
\end{spacing}
\begin{figure}[H]
	\centering
		\includegraphics[width=0.40\textwidth]{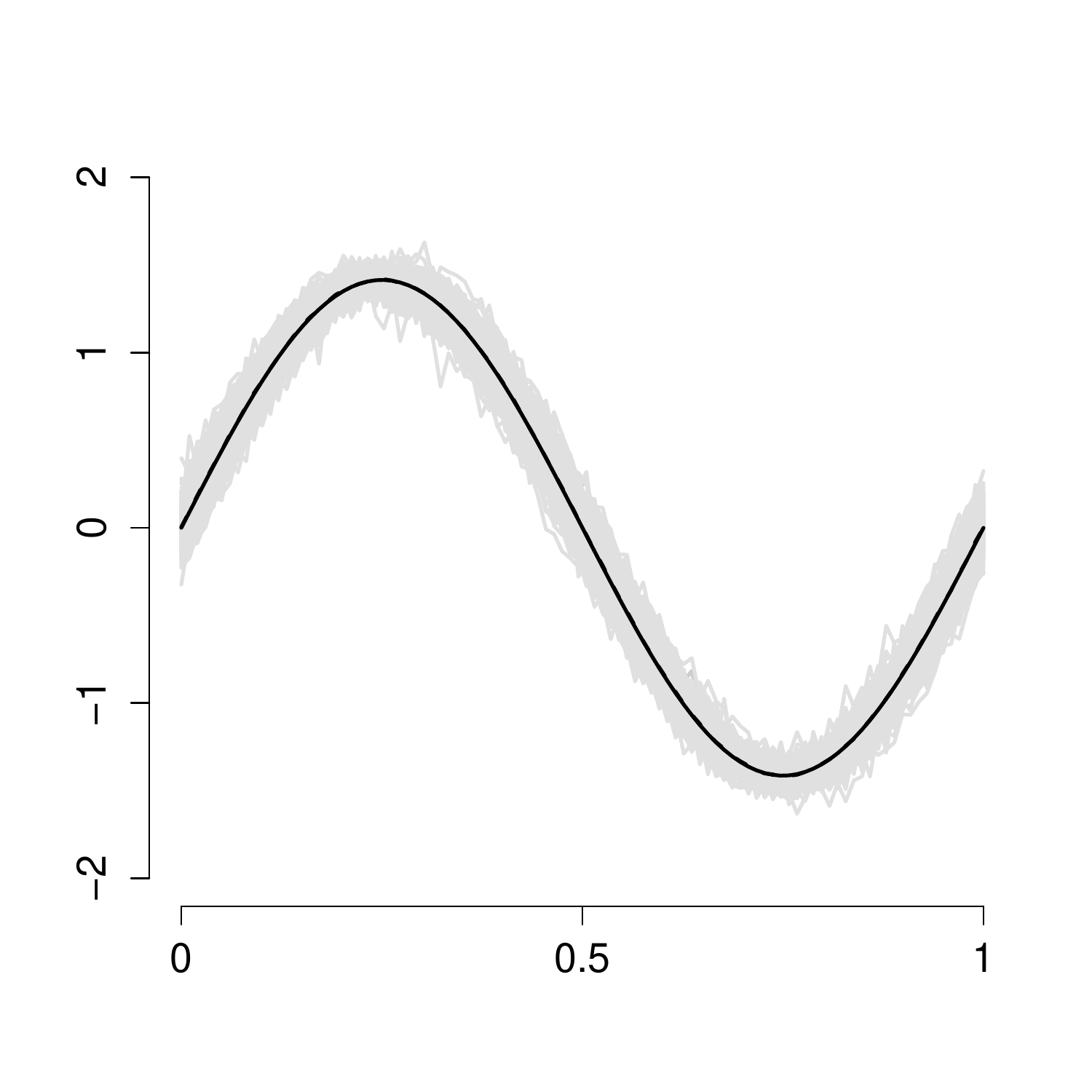}\includegraphics[width=0.40\textwidth]{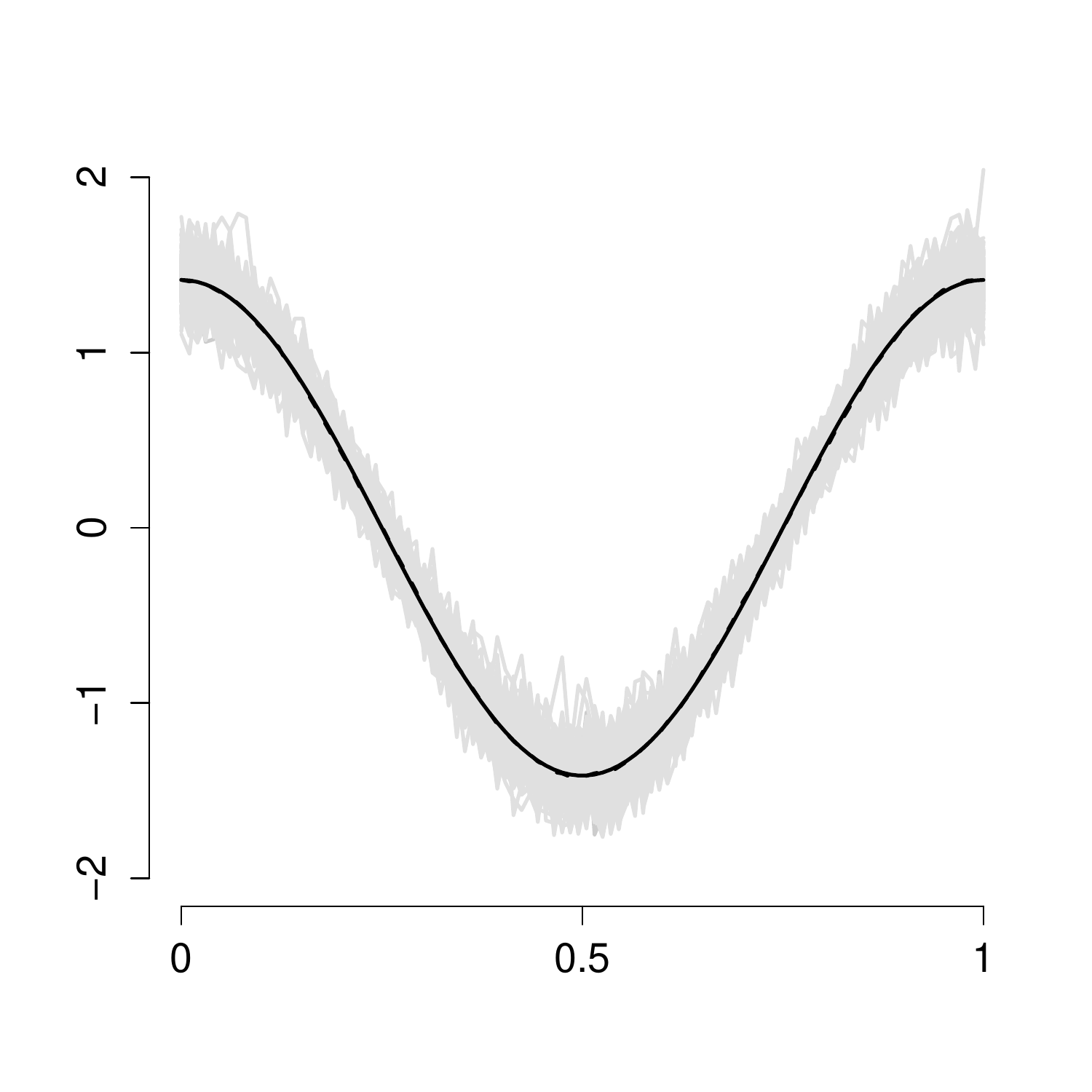}
		\includegraphics[width=0.40\textwidth]{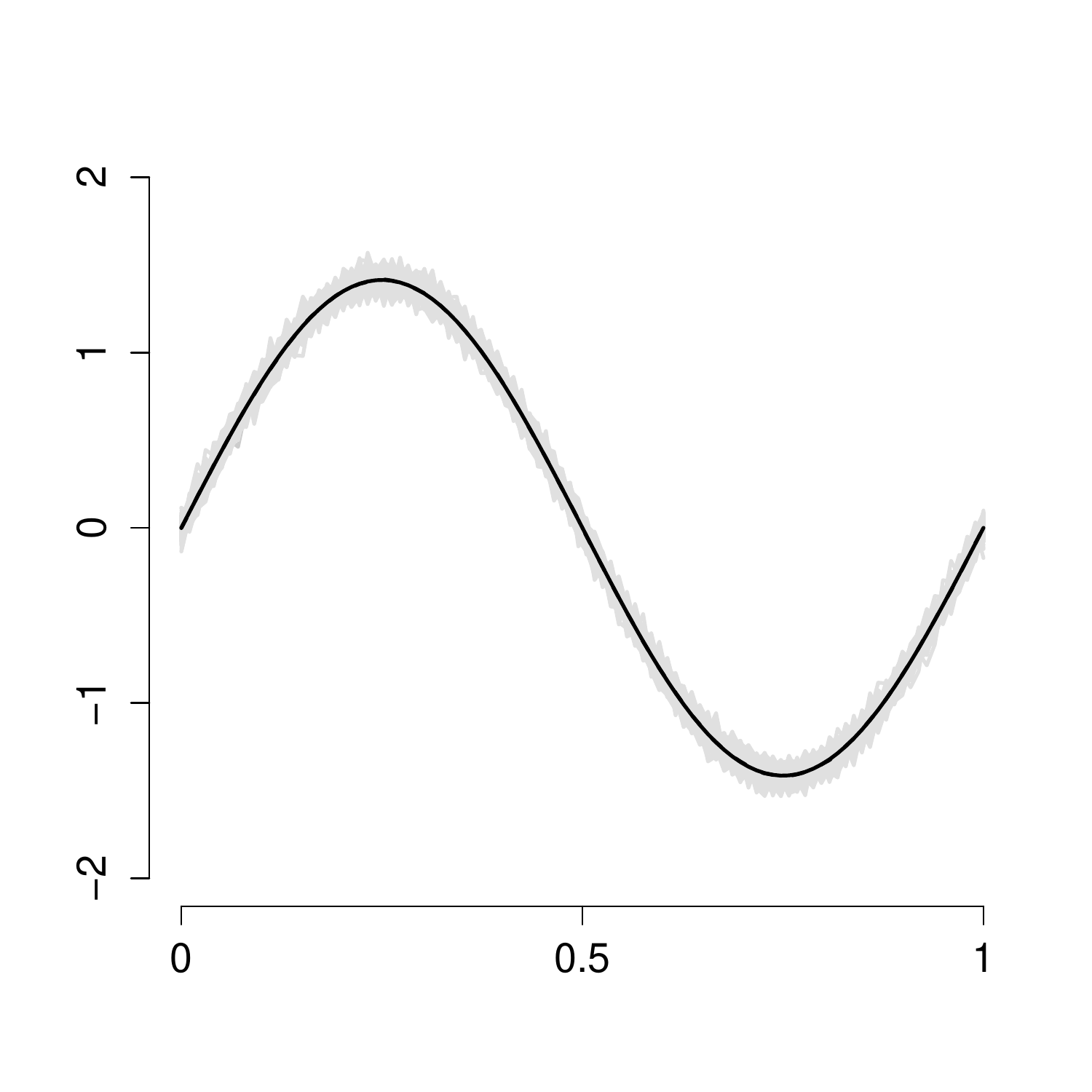}\includegraphics[width=0.40\textwidth]{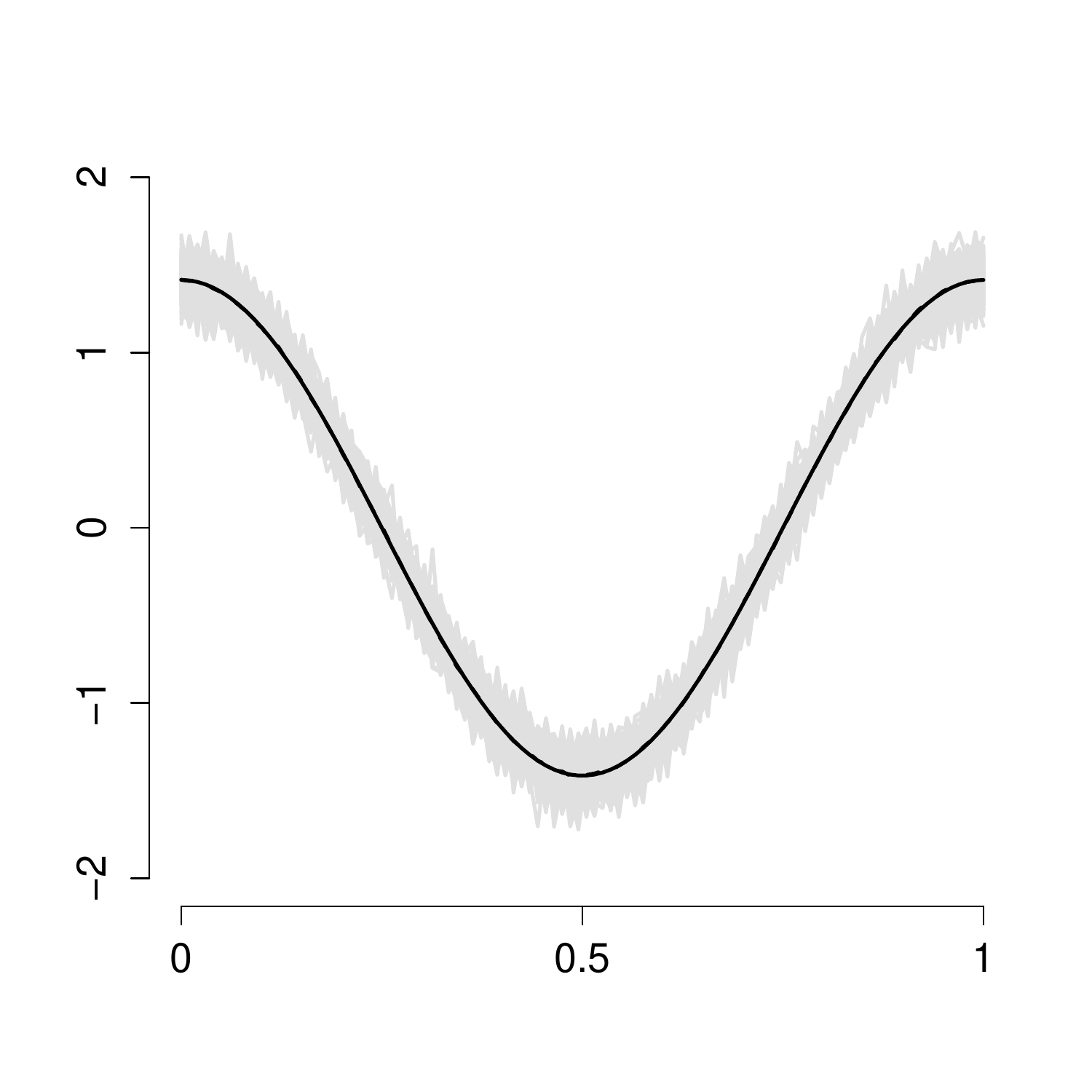}
		\includegraphics[width=0.40\textwidth]{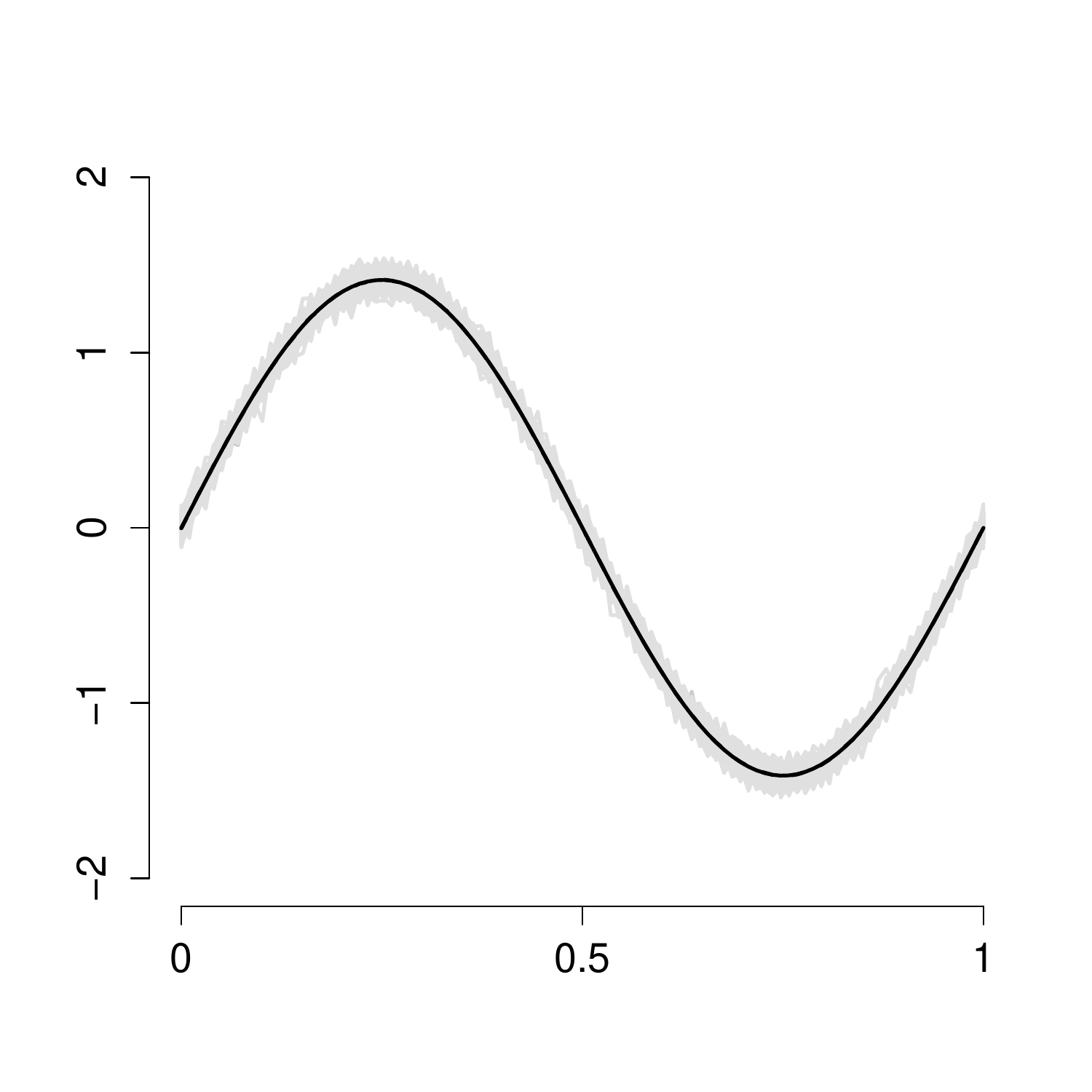}\includegraphics[width=0.40\textwidth]{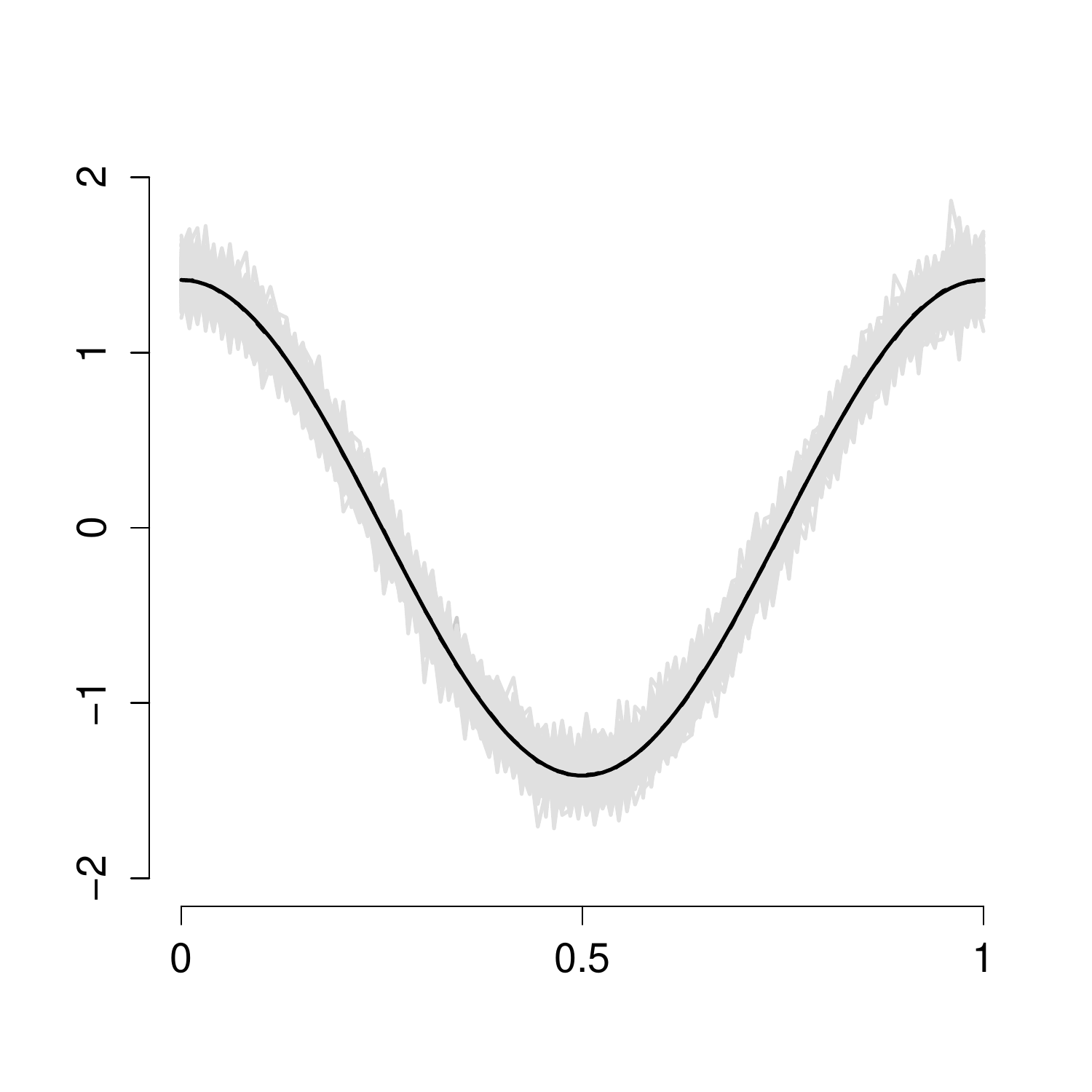}
	\caption{Estimated component functions (solid gray) by 500 simulation runs for simulation setting 1 scenario 1 small sample size and 95\% expectile. The rows from the top to the bottom show respectively results produced by BUP, TD and PEC. Left panel corresponds to the first component function, right panel - to the second. The true functions are shown as solid black curves. The overall mean across simulation runs is shown as dashed black curve. The later can not be distinguished from the true curve.}\label{figpc1}
\end{figure}
\begin{figure}[H]
	\centering
		\includegraphics[width=0.40\textwidth]{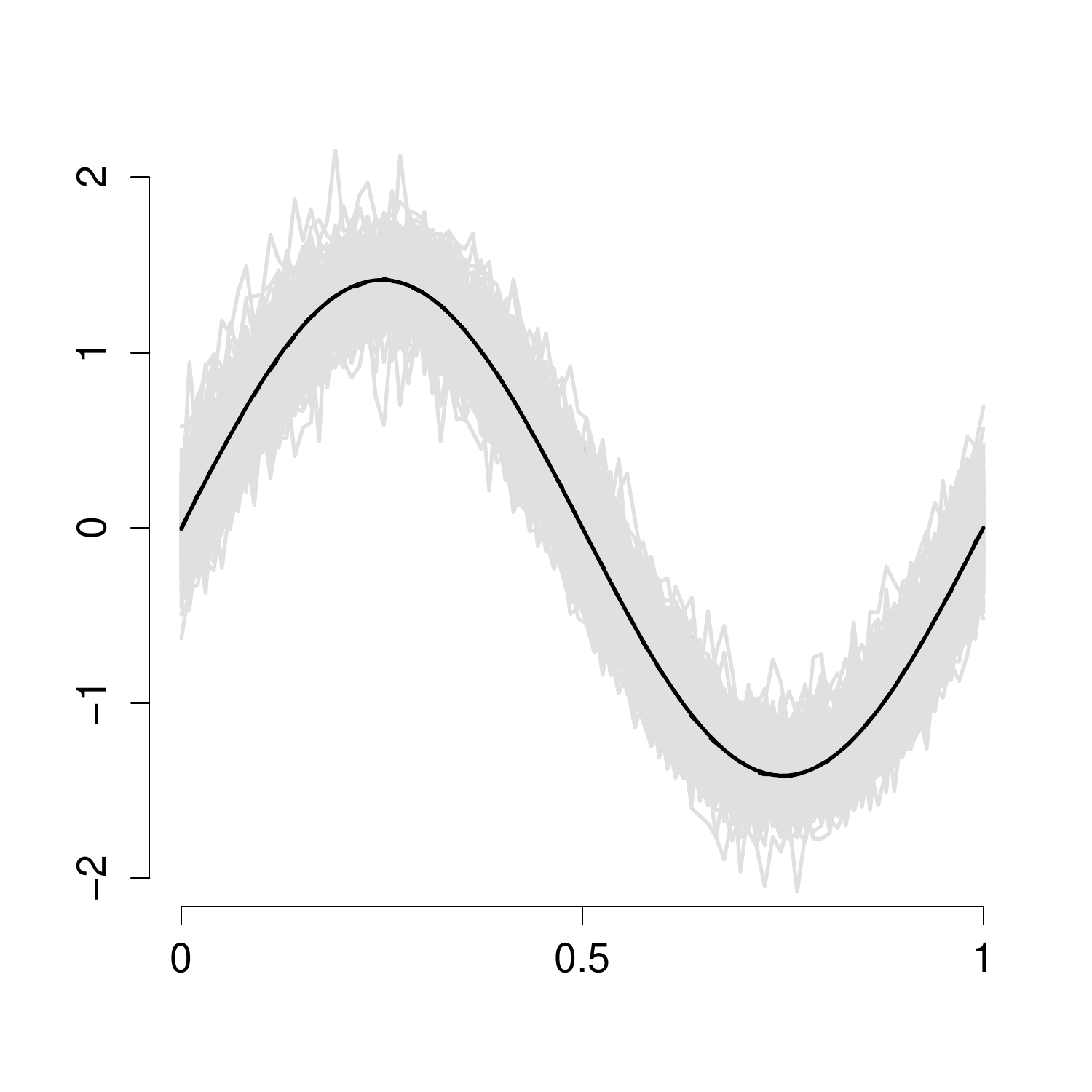}\includegraphics[width=0.40\textwidth]{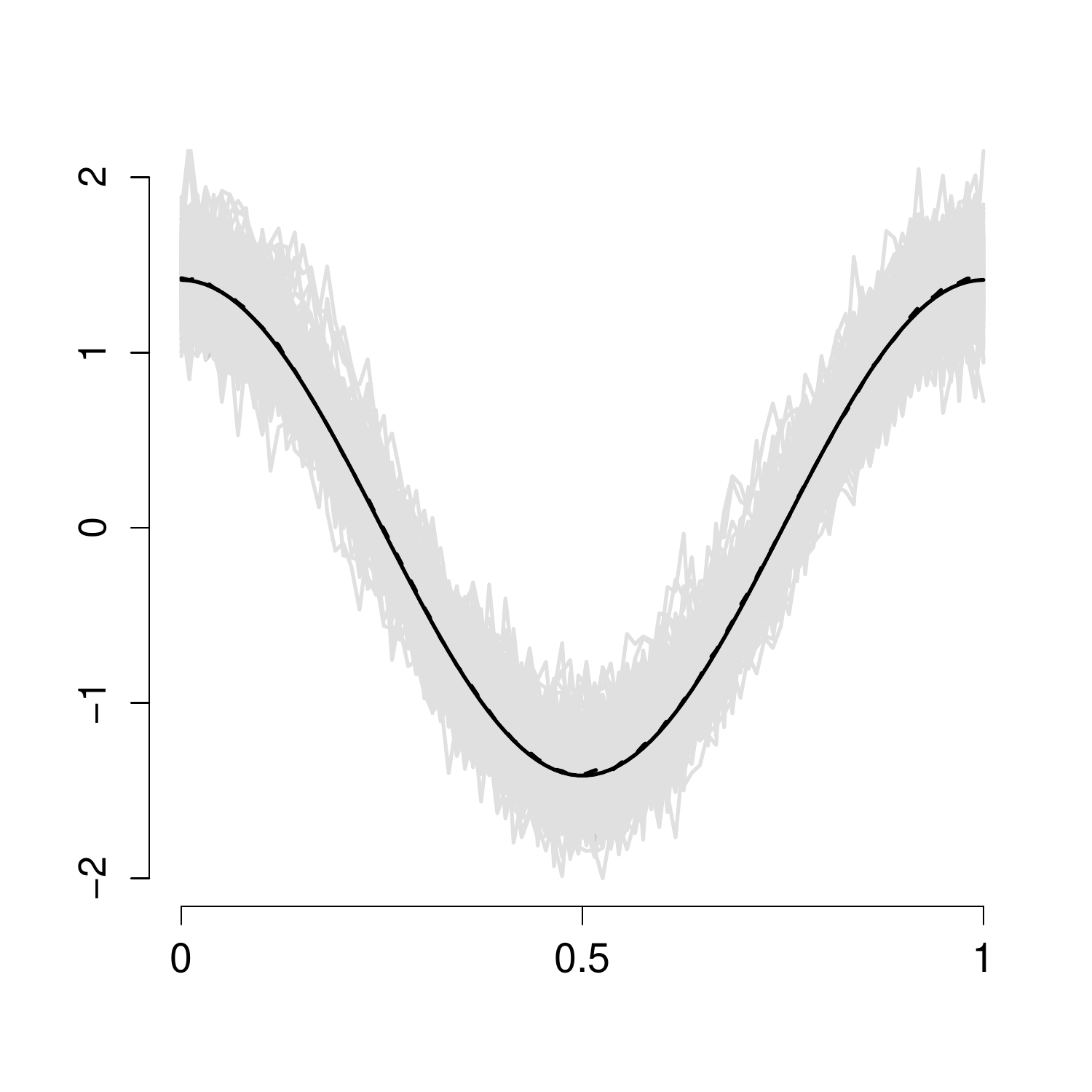}
		\includegraphics[width=0.40\textwidth]{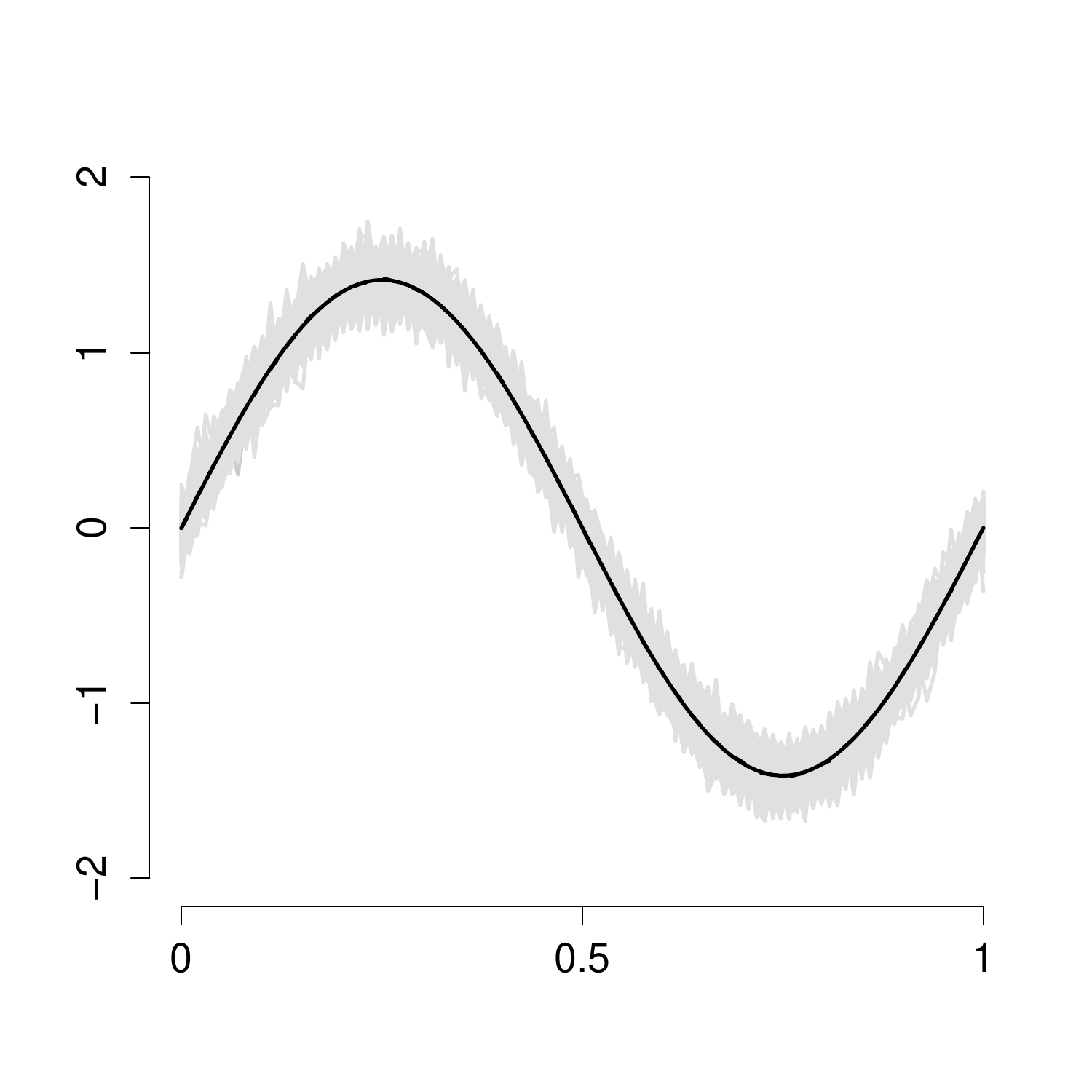}\includegraphics[width=0.40\textwidth]{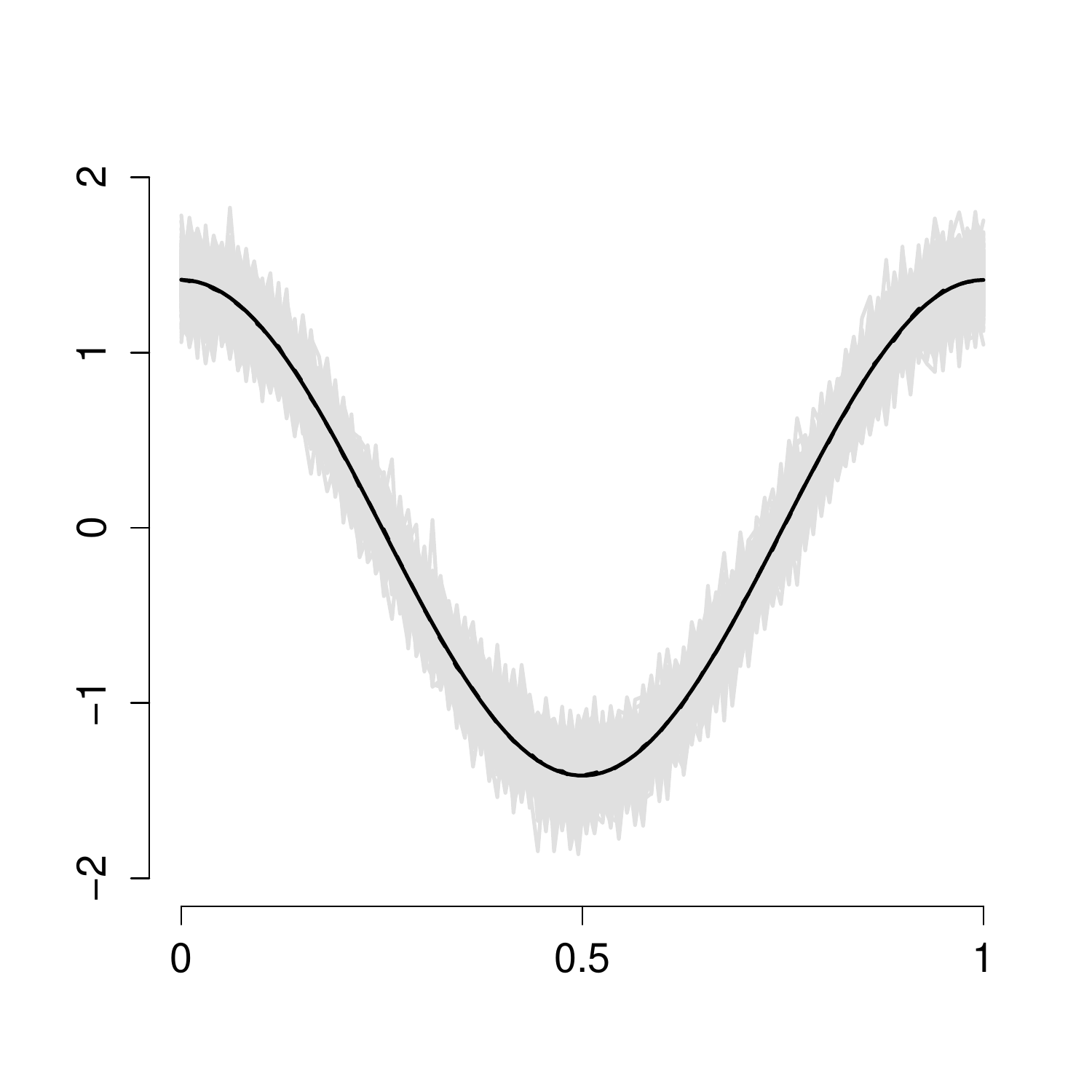}
		\includegraphics[width=0.40\textwidth]{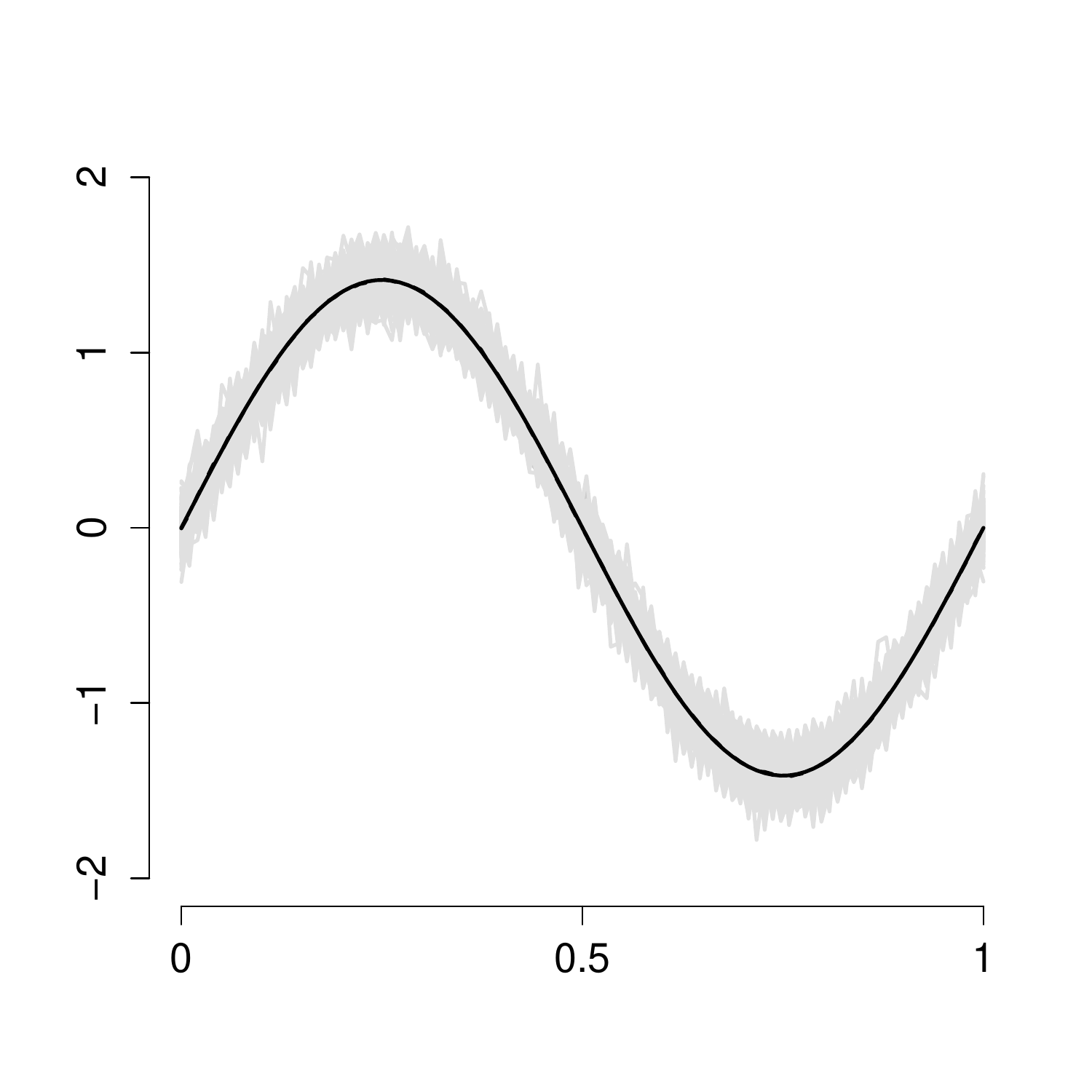}\includegraphics[width=0.40\textwidth]{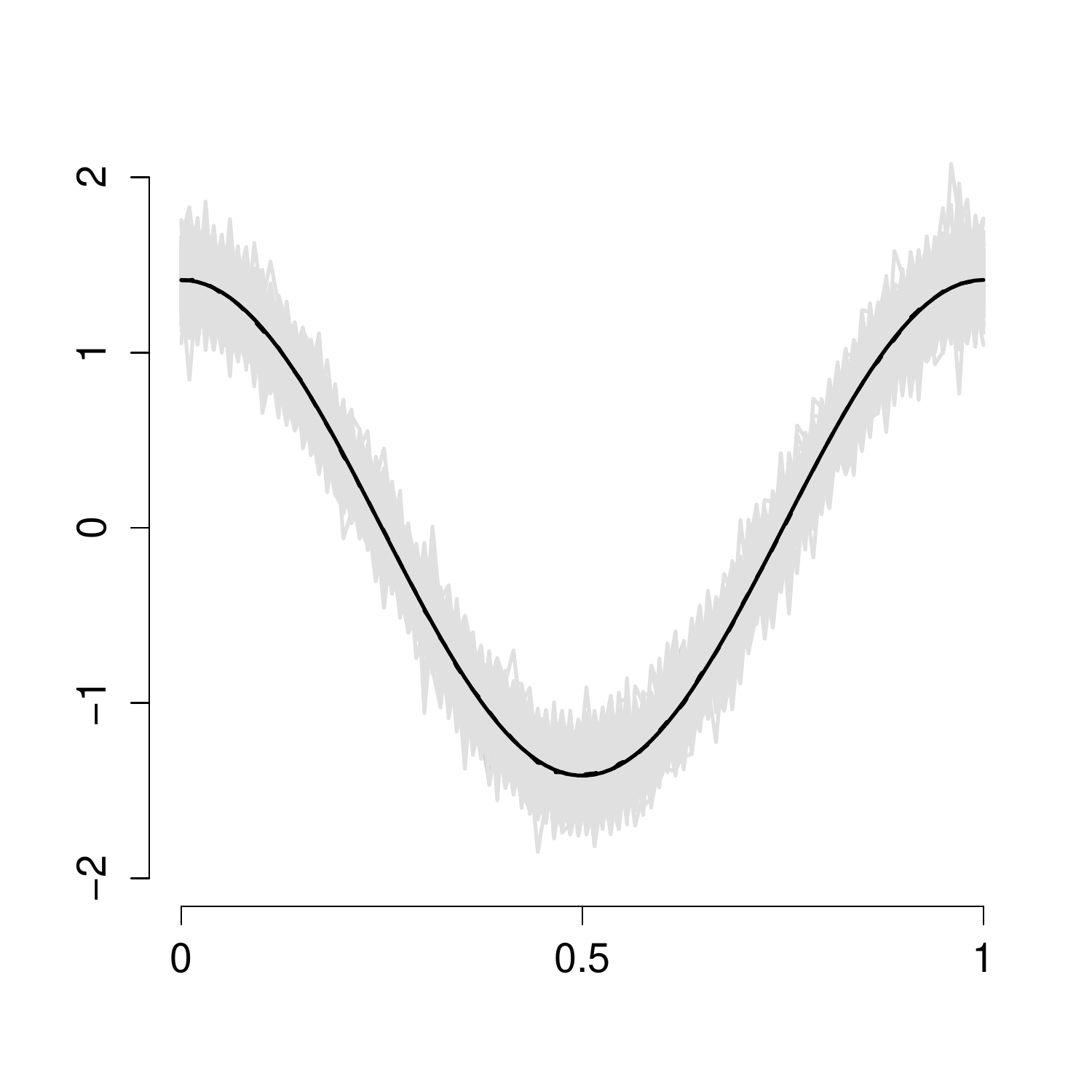}
	\caption{Estimated component functions (gray) by 500 simulation runs for simulation setting 2 scenario 1 small sample size and 95\% expectile. The rows from the top to the bottom show respectively results produced by BUP, TD and PEC. Left panel corresponds to the first component function, right panel - to the second. The true functions are shown as solid black curves. The overall mean across simulation runs is shown as dashed black curve. The later can not be distinguished from the true curve.}\label{figpc2}
\end{figure}
\section{Application to Chinese Weather Data}\label{sec:app}
We apply the algorithms BottomUp, TopDown and PrincipalExpectile to Chinese temperature data with a view to pricing weather derivatives (WDs). WDs are financial instruments written on weather indices as underlyings and are designed to trade with weather related risks. Temperature derivatives are WDs written on a temperature index such as the average temperature recorded at a prespecified weather station. As for  financial derivatives risk factors of temperature are at the core of temperature derivative pricing. In this section we study the risk factors of temperature using daily average temperature data of 159 weather stations in China for the years 1957 to 2009 provided by Chinese Meteorological Administration via its website. We refer to this dataset as the Chinese weather dataset.

To conduct the analysis of the temperature risk factors which are relevant for pricing temperature derivatives we follow the well established methodology of \cite{benth2007}. That is, let $T_{it}$ denote the average temperature at station $i$, $i=1,2,\ldots,n$ in time $t$. We consider each station $i$ separately and using the whole time series of the average temperatures from 1957 to 2009 we fit the following model:
\begin{align}
T_{it}&=X_{it}+\Lambda_{it}\label{tmod}
\end{align}
In (\ref{tmod}) $\Lambda_{it}$ is a seasonal function:
\begin{align*}\Lambda_{it}&=a_i+b_it+\sum_{i=1}^2c_i\sin(2\pi t/365 i)+d_i\cos(2\pi t/365 i)
\end{align*}
and $X_{it}$ is an autoregressive process:
\begin{align*}X_{it}&=\sum_{j=1}^{10}\beta_{ij}X_{i,t-j}+\varepsilon_{it}\nonumber
\end{align*}
We fit the model (\ref{tmod}) to the temperature data of 159 stations and obtain the estimated residuals $\hat{\varepsilon}_{it}$. 

It is crucial to study these risk factors $\hat{\varepsilon}_{it}$ for pricing WDs since the later relies heavily on the distributional properties of ${\varepsilon}_{it}$; frequently ${\varepsilon}_{it}$ are assumed to be Gaussian, \citet{benth2007} and \citet{alaton}. 
The findings of \citet{cadi} reveal the importance of modeling conditional variances beside conditional means to capture the distributional features of temperature. We go beyond this and look at the scale factors in the tails of the ${\varepsilon}_{it}$'s.

To eliminate possible year-specific level and scale effects in ${\varepsilon}_{it}$ for different years, we average and demean the $\hat{\varepsilon}_{it}$ day-wise (the 29th February was droped from the data) over all years, and presmooth them using 23 Fourier series. 

We run the algorithms to estimate a collection of 159 expectile curves for the weather stations at each of the levels 5\%, 50\% and 95\% with respect to days of a year from 1 to 365. Our analysis for the 50\% expectile corresponds to the classical PCA. We estimate first two principal component functions. As we show in Table \ref{varprop} they already explain large portion of the sample variation.
\begin{figure}[H]
	\centering
		\includegraphics[width=0.45\textwidth]{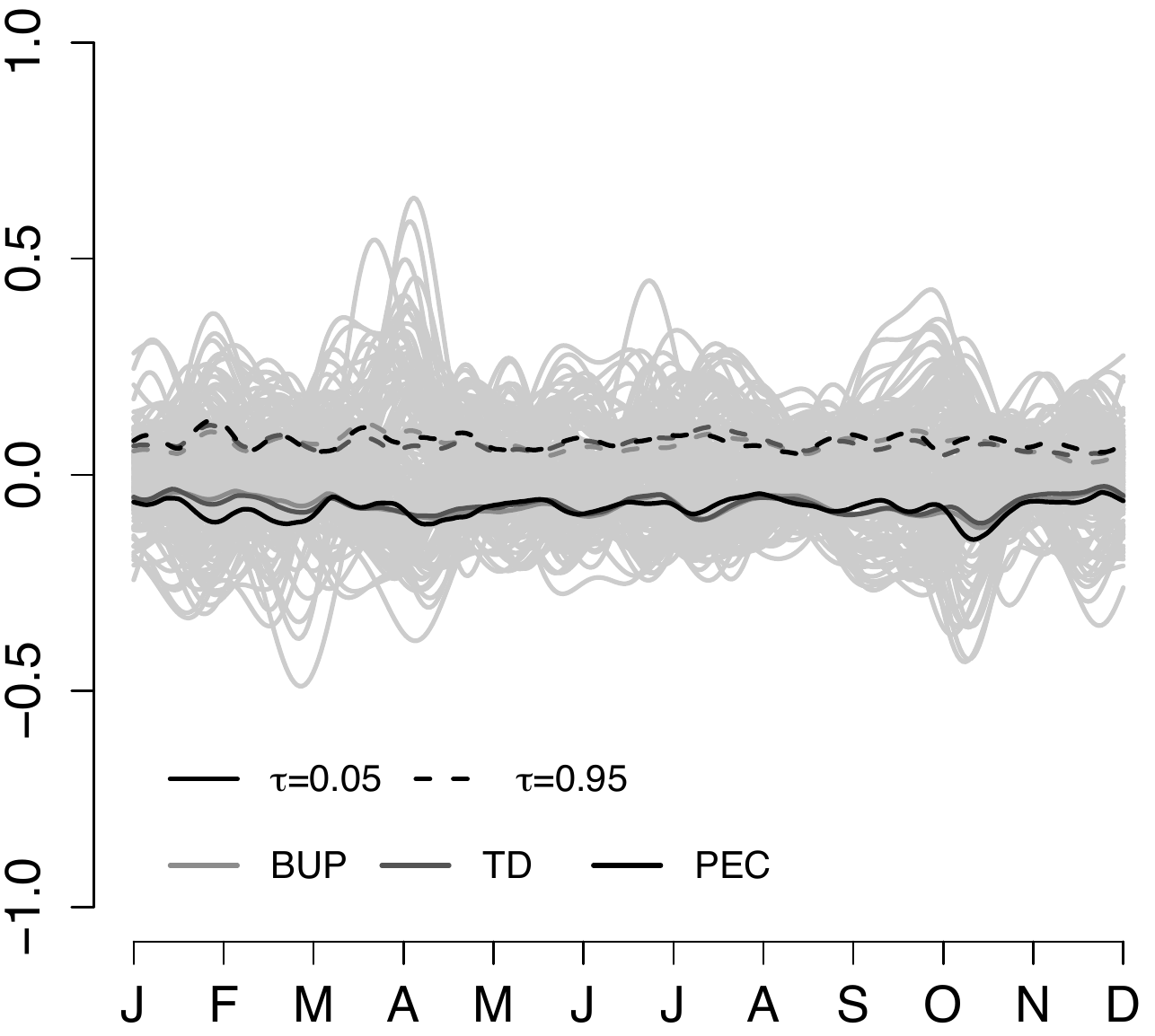}
		\caption{Averaged and smoothed residuals of temperature on 159 stations (gray) and the estimated constants by the algorithms. The horizontal axis features the months from January to December.}
		\label{res_temp}
\end{figure}
The estimation results of the three proposed algorithms are rather similar. On Figures \ref{res_temp} and \ref{res_temp1} we present the estimated constant terms  and the estimated principal component functions for $\tau=0.05$ and $\tau=0.95$. 
\begin{figure}[H]
	\centering
		\includegraphics[width=0.45\textwidth]{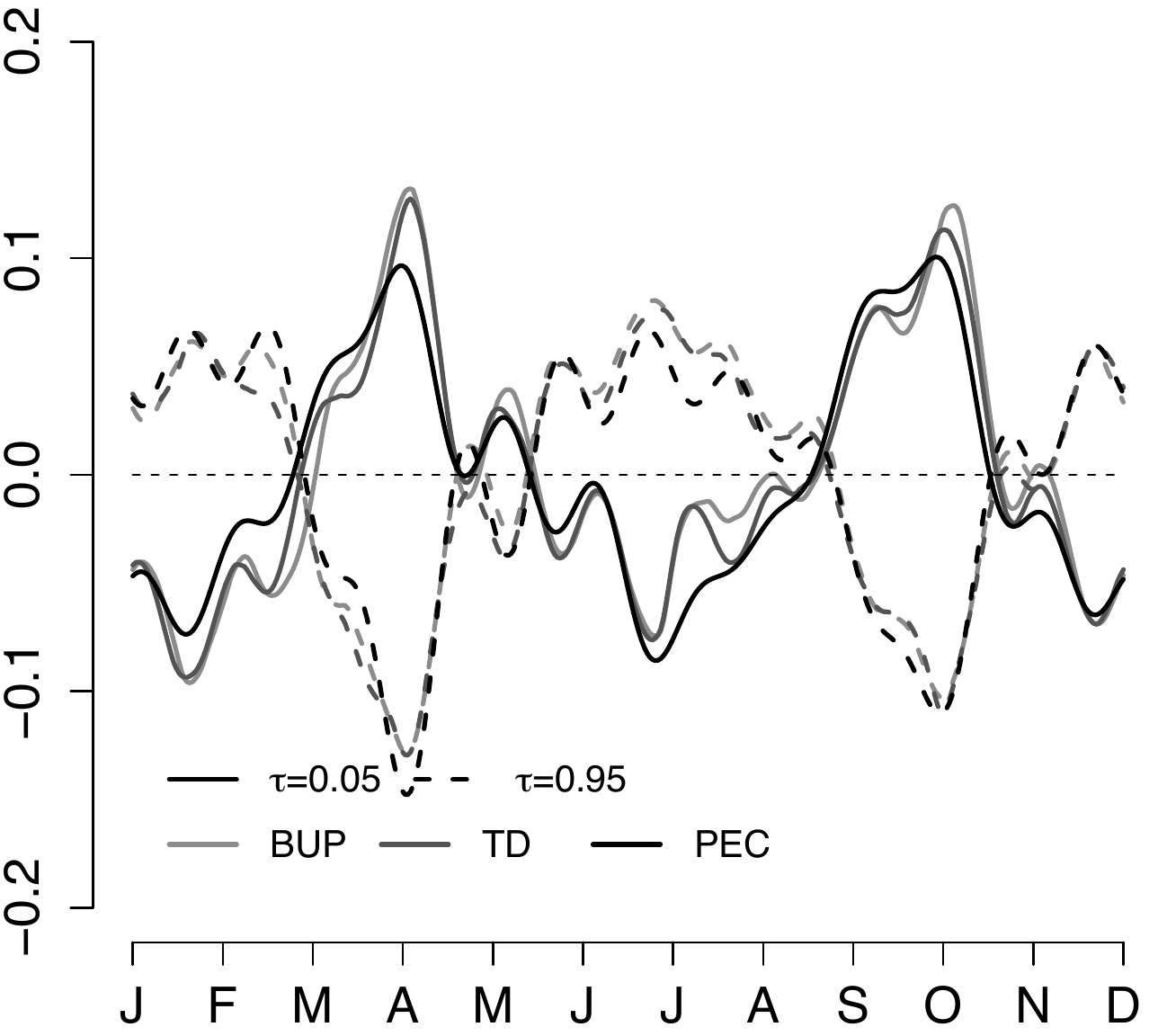} \includegraphics[width=0.45\textwidth]{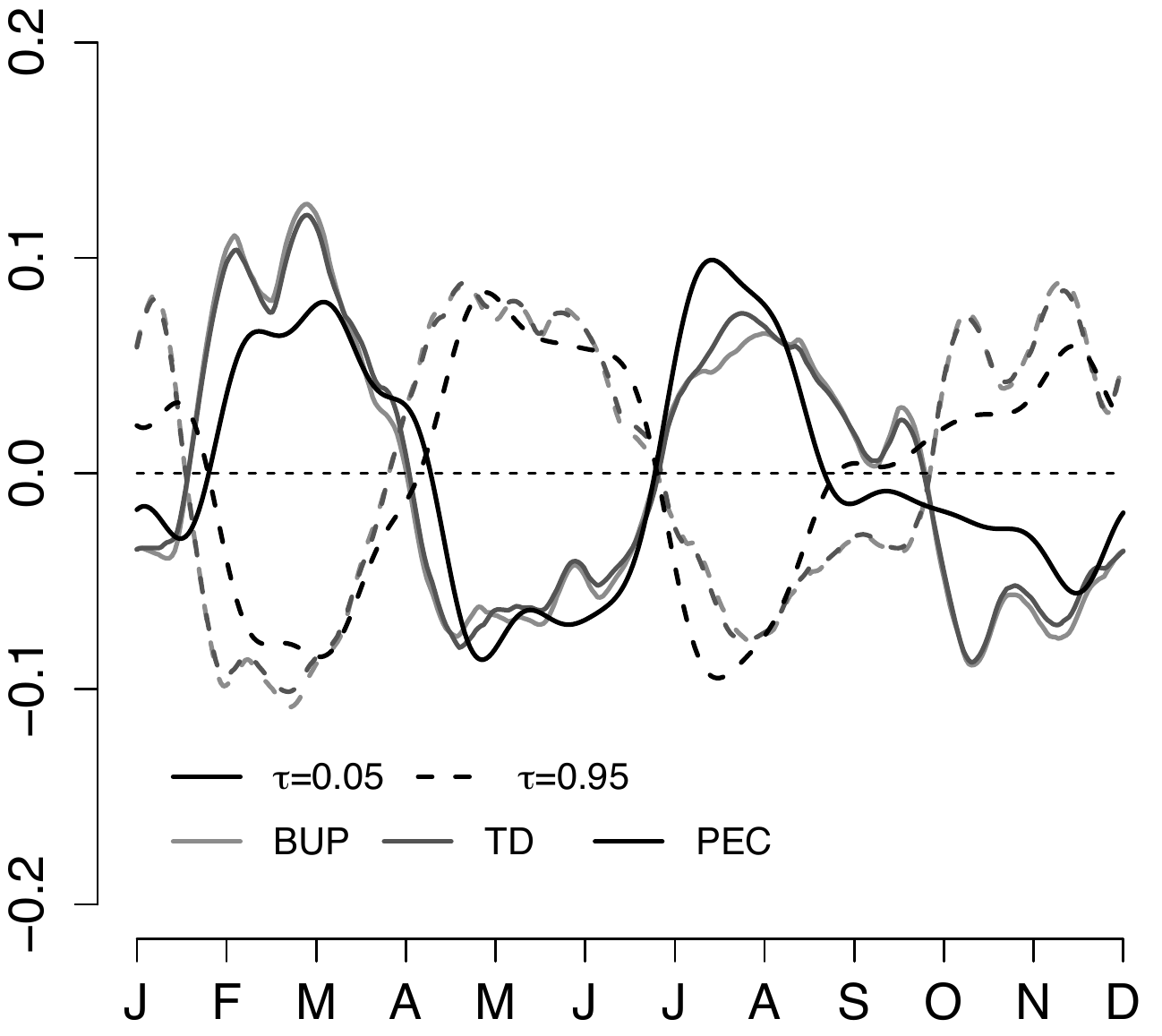}
\caption{Left: the estimated first component function for the residuals of temperature. Right: the estimated second component function.}
		\label{res_temp1}
\end{figure}
The figures reveal that i. the estimators for the constant term produced by different algorithms are rather close to each other; ii. the estimators of the principal component functions returned by the algorithms are also quite similar; iii. BUP and TD principal components are particularly close to each other. 
\begin{table}[H]
\centering
\begin{tabular}{crrrr}\hline\hline
&&BUP& TD& PEC \\\hline
$\tau=$&0.05&0.89&0.89&0.86\\
$\tau=$&0.50&0.65&0.65&0.65\\
$\tau=$&0.95&0.89&0.90&0.87\\\hline\hline
\end{tabular}
\caption{Proportion of the explained variance by the two PCs for different $\tau$-levels and each of the algorithms in the temperature residuals curves}\label{varprop}
\end{table}
The obtained first and second components indicate changes in the temperature distribution from lighter to heavier tails and the other way around within a typical year. A positive score on the first component would mean lighter than average tails of the temperature distribution in spring and fall, and heavier than average tails in winter and summer. Similar, a positive score on the second component would indicate lighter than average tails of the temperature distribution in February, March, April, July, August, and September, and heavier than average tails during the rest of the year. 
\begin{figure}[H]
	\centering
		\includegraphics[width=0.45\textwidth]{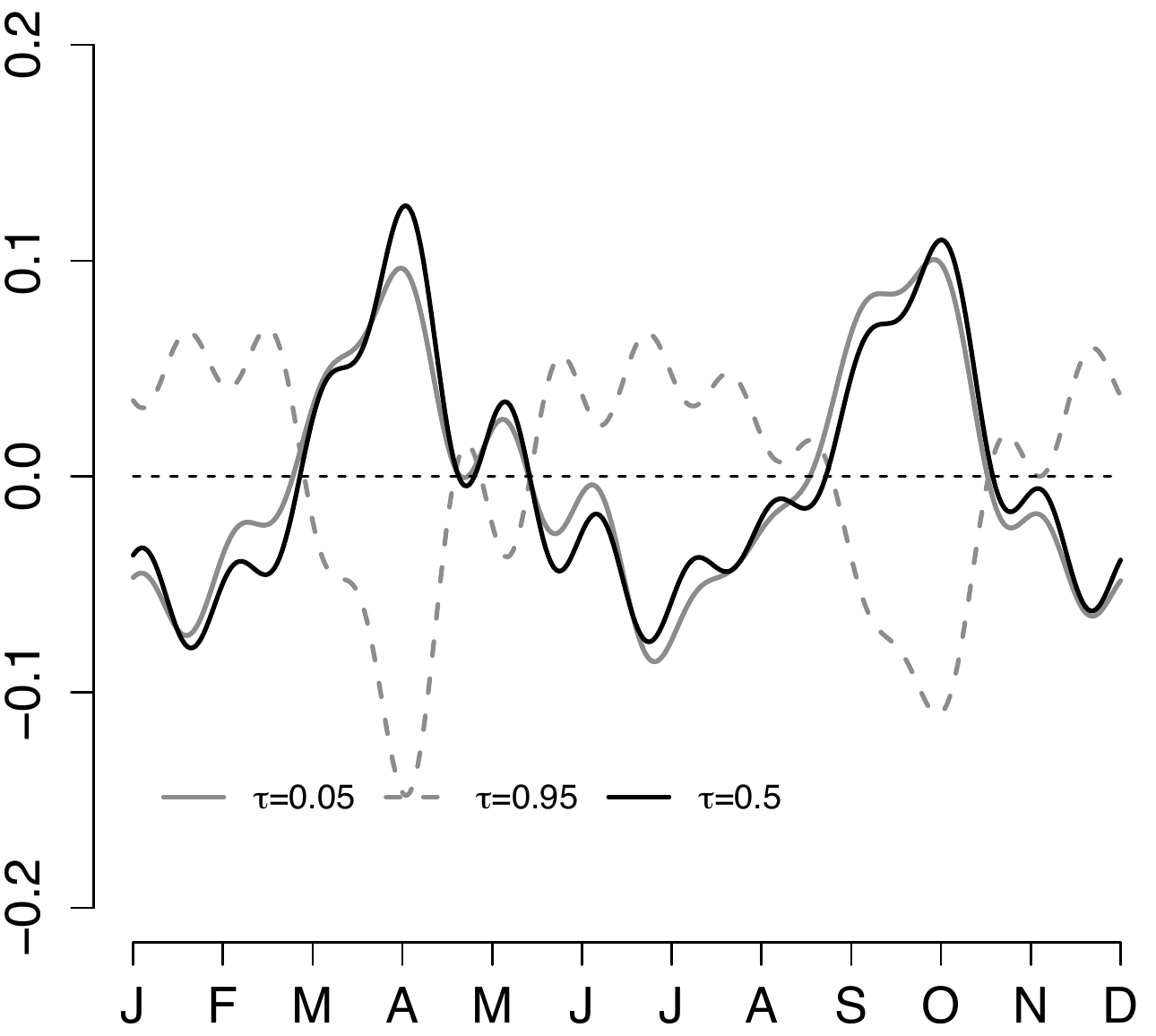} \includegraphics[width=0.45\textwidth]{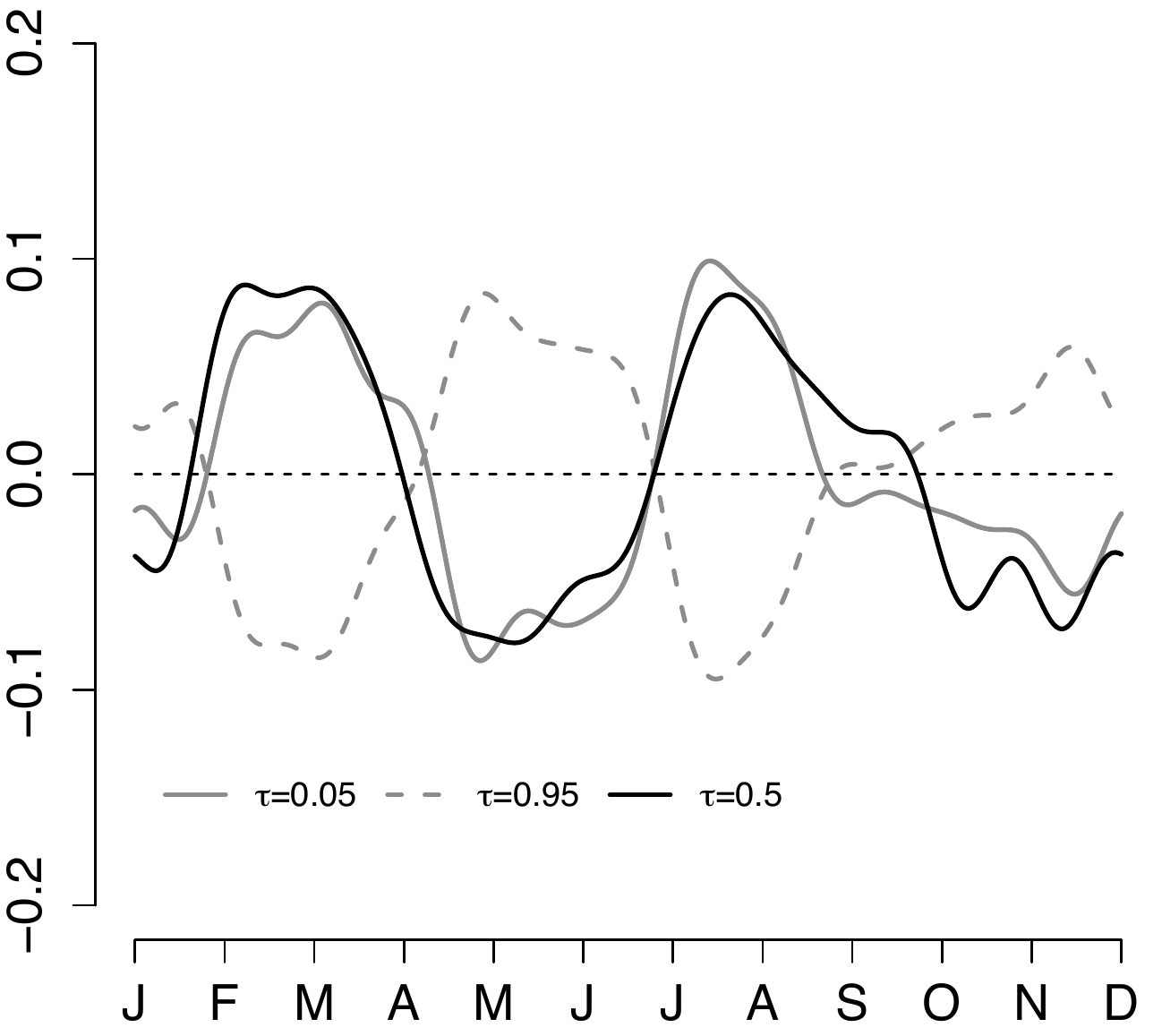}
\caption{Left panel: the estimated first PEC for $\tau=\{0.05,0.5,0.95\}$. Right panel: the estimated second PEC for $\tau=\{0.05,0.5,0.95\}$.}
		\label{pec_pc}
\end{figure}
In Figure \ref{pec_pc} we show the principal component functions for $\tau$= 0.05, 0.5, and 0.95 obtained by PrincipalExpectile. We observe that the estimated principal component functions vary with $\tau$ and exhibit some differences to the classical PCA where $\tau=0.5$. By applying Proposition \ref{proppec}(3) to PEC, we conclude that the distribution of the considered temperature residuals is rather not an elliptically symmetric one. Thus, the normality assumption for pricing WDs on temperature as needed in the technology presented by \cite{benth2007} might be violated for this data.

%% file: summary.tex
We proposed two definitions of principal components in an asymmetric norm and provided consistent algorithms based on iterative least squares. We derived the upper bounds on their convergence times as well as other useful properties of the resulting principal components in an asymmetric norm. 

The algorithms TopDown and BottomUp minimize the projection error in a $\tau$-asym\-met\-ric norm, and PrincipalExpectile algorithm maximizes the $\tau$-variance of the low-dimensional projection. The later algorithm was shown to share 'nice' properties of PCA as invariance under translations and changes of basis, moreover, it coincides with classical PCA for elliptically symmetric distributions. 

Using simulations we compared finite sample performance of the proposed algorithms. All algorithms appear to produce similar results. Overall performance of PrincipalExpectile and TopDown was very satisfactory in terms of the MSE, PrincipalExpectile showed robustness to 'fat-tails' and skewness of the data distribution. 

We applied the algorithms to a Chinese weather dataset with a view to weather derivative pricing. Using a commonly accepted model for temperature of \cite{benth2007}, we estimated the first two principal component functions of the temperature residuals as functions of days of a year. The resulting component functions indicate relative changes in the tails of the temperature distribution from light to heavier and vice versa. Our further results question the validity of the normality assumption on the temperature residuals which is frequently used for pricing temperature based derivatives.

The proposed algorithms appear to be a good way to study extremes of multivariate data. They are easy to compute, relatively fast and their results are easy to interpret.

%% file: paper20140114.bbl
\begin{thebibliography}{22}
\newcommand{\enquote}[1]{``#1''}
\expandafter\ifx\csname natexlab\endcsname\relax\def\natexlab#1{#1}\fi

\bibitem[\protect\citeauthoryear{Alaton, Djehiche, and Stillberger}{Alaton
  et~al.}{2002}]{alaton}
\textsc{Alaton, P., B.~Djehiche, and D.~Stillberger} (2002): \enquote{On
  Modelling and Pricing Weather Derivatives,} \emph{Applied Mathematical
  Finance}, 1, 1--20.

\bibitem[\protect\citeauthoryear{Benth, Benth, and Koekebakker}{Benth
  et~al.}{2007}]{benth2007}
\textsc{Benth, F., J.~Benth, and S.~Koekebakker} (2007): \enquote{Putting a
  price on temperature,} \emph{Scandinavian Journal of Statistics}, 34,
  746--767.

\bibitem[\protect\citeauthoryear{Benth and Benth}{Benth and
  Benth}{2012}]{benth}
\textsc{Benth, J. and F.~Benth} (2012): \enquote{A critical view on temperature
  modelling for application in weather derivatives markets,} \emph{Energy
  Economics}, 34, 592--602.

\bibitem[\protect\citeauthoryear{Campbell and Diebold}{Campbell and
  Diebold}{2005}]{cadi}
\textsc{Campbell, S. and F.~Diebold} (2005): \enquote{Weather forecasting for
  weather derivatives,} \emph{Journal of the American Statistical Association},
  100, 6--16.

\bibitem[\protect\citeauthoryear{Chen and M\"uller}{Chen and
  M\"uller}{2012}]{chenmuel}
\textsc{Chen, K. and H.-G. M\"uller} (2012): \enquote{Conditional Quantile
  analysis when covariates are functions, with application to growth data,}
  \emph{Journal of the Royal Statistical Society: Series B (Statistical
  Methodology)}, 874, 67--89.

\bibitem[\protect\citeauthoryear{Cobza{\c{s}}}{Cobza{\c{s}}}{2013}]{asym}
\textsc{Cobza{\c{s}}, {\c{S}}.} (2013): \emph{Functional analysis in asymmetric
  normed spaces}, Springer.

\bibitem[\protect\citeauthoryear{Crambes, Kneip, and P.}{Crambes
  et~al.}{2009}]{crknsa}
\textsc{Crambes, C., A.~Kneip, and S.~P.} (2009): \enquote{Smooth splines
  estimators for functional linear regression,} \emph{Annals of Statistics},
  37, 35--72.

\bibitem[\protect\citeauthoryear{Doss and Gill}{Doss and Gill}{1992}]{dossgill}
\textsc{Doss, H. and R.~Gill} (1992): \enquote{An Elementary Approach to Weak
  Convergence for Quantile Processes, With Applications to Censored Survival
  Data,} \emph{Journal of the American Statistical Association}, 87, pp.
  869--877.

\bibitem[\protect\citeauthoryear{Fraiman and Pateiro-L{\'o}pez}{Fraiman and
  Pateiro-L{\'o}pez}{2012}]{fraimanlopez}
\textsc{Fraiman, R. and B.~Pateiro-L{\'o}pez} (2012): \enquote{Quantiles for
  finite and infinite dimensional data,} \emph{Journal of Multivariate
  Analysis}, 108, 1--14.

\bibitem[\protect\citeauthoryear{Guo, Zhou, H{\"a}rdle, and Huang}{Guo
  et~al.}{2013}]{gzhh}
\textsc{Guo, M., L.~Zhou, W.~H{\"a}rdle, and J.~Huang} (2013):
  \enquote{Functional Data Analysis for Generalized Quantile Regression,}
  \emph{Statistics and Computing}, doi: 10.1007/s11222-013-9425-1, 1--14.

\bibitem[\protect\citeauthoryear{H\"ardle and L\'opez~Cabrera}{H\"ardle and
  L\'opez~Cabrera}{2012}]{haelop}
\textsc{H\"ardle, W. and B.~L\'opez~Cabrera} (2012): \enquote{The Implied
  Market Price of Weather Risk,} \emph{Applied Mathematical Finance}, 19,
  59--95.

\bibitem[\protect\citeauthoryear{Hjort and Pollard}{Hjort and
  Pollard}{2011}]{hjort2011asymptotics}
\textsc{Hjort, N. and D.~Pollard} (2011): \enquote{Asymptotics for minimisers
  of convex processes,} \emph{arXiv preprint arXiv:1107.3806}.

\bibitem[\protect\citeauthoryear{Jolliffe}{Jolliffe}{2004}]{j04}
\textsc{Jolliffe, I.} (2004): \emph{Principal component analysis}, Springer.

\bibitem[\protect\citeauthoryear{Kneip and Utikal}{Kneip and
  Utikal}{2001}]{kneip}
\textsc{Kneip, A. and K.~Utikal} (2001): \enquote{Inference for Density
  Families Using Functional Principal Component Analysis,} \emph{Journal of the
  American Statistical Association}, 96, 519--532.

\bibitem[\protect\citeauthoryear{Kong and Mizera}{Kong and
  Mizera}{2012}]{kongmizera}
\textsc{Kong, L. and I.~Mizera} (2012): \enquote{Quantile tomography: using
  quantiles with multivariate data,} \emph{Statistica Sinica}, 22, 1589--1610.

\bibitem[\protect\citeauthoryear{Kuan, Yeh, and Hsu}{Kuan et~al.}{2009}]{kuan}
\textsc{Kuan, C.-M., J.-H. Yeh, and Y.-C. Hsu} (2009): \enquote{Assessing value
  at risk with CARE, the Conditional Autoregressive Expectile models,}
  \emph{Journal of Econometrics}, 150, 261--270.

\bibitem[\protect\citeauthoryear{Newey and Powell}{Newey and Powell}{1987}]{np}
\textsc{Newey, W. and J.~Powell} (1987): \enquote{Asymmetric least squares
  estimation and testing,} \emph{Econometrica}, 819--847.

\bibitem[\protect\citeauthoryear{Ramsay and Silverman}{Ramsay and
  Silverman}{2005}]{rs05}
\textsc{Ramsay, J. and B.~Silverman} (2005): \emph{Functional data analysis},
  Springer, New York.

\bibitem[\protect\citeauthoryear{Schnabel}{Schnabel}{2011}]{schnabel}
\textsc{Schnabel, S.} (2011): \enquote{Expectile smoothing: new perspectives on
  asymmetric least squares. An application to life expectancy,} Ph.D. thesis,
  Utrecht University.

\bibitem[\protect\citeauthoryear{Srebro and Jaakkola}{Srebro and
  Jaakkola}{2003}]{srebro}
\textsc{Srebro, N. and T.~Jaakkola} (2003): \enquote{Weighted low-rank
  approximations,} in \emph{Machine Learning International Workshop}, vol.~20,
  720.

\bibitem[\protect\citeauthoryear{Taylor}{Taylor}{2008}]{tailor2008}
\textsc{Taylor, J.} (2008): \enquote{Estimating Value at Risk and Expected
  Shortfall Using Expectiles,} \emph{Journal of Financial Econometrics}, 6,
  231--252.

\bibitem[\protect\citeauthoryear{van~der Vaart and Wellner}{van~der Vaart and
  Wellner}{1996}]{pollard}
\textsc{van~der Vaart, A. and J.~Wellner} (1996): \emph{Weak Convergence and
  Empirical Processes: With Applications to Statistics}, Springer Series in
  Statistics, Springer.

\end{thebibliography}
